\documentclass[preprint,aps,onecolumn,preprintnumbers,amsmath,amssymb,superscriptaddress,nofootinbib]{revtex4-1}
\usepackage{graphicx}
\usepackage{subfigure}
\usepackage{epsfig}
\usepackage{dcolumn}
\usepackage{bm}
\usepackage{ulem}
\usepackage{color}
\usepackage{multirow} 
\usepackage{slashed}
\usepackage{hyperref}
\usepackage{amsthm}

\def\bra#1{\left\langle#1\right|}
\def\ket#1{\left|#1\right\rangle}

\def\kc#1{\left(#1\right)}
\def\kd#1{\left[#1\right]}
\def\ke#1{\left\{#1\right\}}

\def\be{\begin{equation}}       \def\ee{\end{equation}}
\def\bea{\begin{eqnarray}}      \def\eea{\end{eqnarray}}
\def\ba{\begin{array}}
\def\ea{\end{array}}
\def\bnum{\begin{enumerate} }
\def\enum{\end{enumerate}}
\def\bpm{\begin{pmatrix}[0.5]}
\def\epm{\end{pmatrix}}
\def\bbm{\begin{bmatrix}[0.5]}
\def\ebm{\end{bmatrix}}

\def\=>{\Rightarrow}
\def\>{\rightarrow}

\def\eye2{Fathbb{I}}

\def\m{\textrm{matter}}

\def\Tr{\mathrm{Tr}}

\renewcommand{\>}{\rangle}

\newtheorem{definition}{Definition}
\newtheorem{theorem}{Theorem}
\newtheorem{lemma}{Lemma}

\makeatletter
\renewcommand*\env@matrix[1][\arraystretch]{
  \edef\arraystretch{#1}
  \hskip -\arraycolsep
  \let\@ifnextchar\new@ifnextchar
  \array{*\c@MaxMatrixCols c}}
\makeatother

\begin{document}

\newcommand{\IHEP}{\affiliation{Institute of High Energy Physics, Chinese Academy of Sciences, Beijing 100049, China}}
\newcommand{\UCAS}{\affiliation{School of Physics, University of Chinese Academy of Sciences, Beijing 100049, China}}
\newcommand{\ITP}{\affiliation{Institute of Theoretical Physics, Chinese Academy of Science, Beijing 100190, China}}

\title{Tensor chain and constraints in tensor networks}
\author{Yi Ling}
\email{lingy@ihep.ac.cn}
\author{Yuxuan Liu}
\email{liuyuxuan@ihep.ac.cn}
\IHEP\UCAS
\author{Zhuo-Yu Xian}
\email{xianzy@itp.ac.cn}
\IHEP\UCAS\ITP
\author{Yikang Xiao}
\email{ykxiao@ihep.ac.cn}
\IHEP\UCAS

\begin{abstract}
    This paper accompanies with our recent work on quantum error
        correction (QEC) and entanglement spectrum (ES) in tensor networks
        (arXiv:1806.05007). We propose a general framework for planar
        tensor network state with tensor constraints as a model for
        $AdS_3/CFT_2$ correspondence, which could be viewed as a
        generalization of hyperinvariant tensor networks recently proposed
        by Evenbly. We elaborate our proposal on tensor chains in a tensor
        network by tiling $H^2$ space and provide a diagrammatical
        description for general multi-tensor constraints in terms of
        tensor chains, which forms a generalized greedy algorithm. The
        behavior of tensor chains under the action of greedy algorithm is
        investigated in detail. In particular, for a given set of tensor
        constraints, a critically protected (CP) tensor chain can be
        figured out and evaluated by its average reduced interior angle.
        We classify tensor networks according to their ability of QEC and
        the flatness of ES. The corresponding geometric description of
        critical protection over the hyperbolic space is also given.

\end{abstract}
\maketitle

\tableofcontents

\section{Introduction}\label{SectionInt}
Tensor network as a powerful tool for building the ground state of
a many-body system has been greatly investigated in recent
    years \cite{Vidal:2007hda}. One remarkable feature of tensor
    network states is the intuitive description of quantum
    entanglement among local degrees of freedom. For a subsystem
    composed of some uncontracted edges in a tensor network, its entanglement entropy is vividly bounded by the minimal cuts disconnecting this subsystem and its complementarity. This scenario can be viewed as
    the discretized description of Ryu-Takayanagi (RT) formula in
    holographic approach \cite{Ryu:2006bv}. Inspired by this, people find that a
    holographic space can emerge from entanglement renormalization of
    a many-body system \cite{Swingle:2009bg,Swingle:2012wq}. It has
    further been conjectured in \cite{VanRaamsdonk:2010pw} and \cite{Maldacena:2001kr} that the
    classical connectivity of spacetime arises by entangling the
    degrees of freedom in two components. As a bridge between quantum
    entanglement and the structure of spacetime, tensor networks have
    been providing a practical framework for exploring the emergence
    of spacetime in the context of gauge/gravity duality
    \cite{Nozaki:2012zj,Qi:2013caa}.

Another property of entanglement enjoyed by holographic duality is
quantum error correction (QEC) \cite{Schumacher:1996dy}. Based on
sub system duality, operators in the bulk can be reconstructed by
the operators supported on a sub system of the boundary
\cite{Hamilton:2006az,Dong:2016eik,Harlow:2018fse,Cotler:2017erl}.
In other words, there are subspaces of the Hilbert space in the
bulk which can still be reconstructed even if an amount of
information on the boundary is erased
\cite{Mintun:2015qda,Freivogel:2016zsb,Harlow:2016vwg,Almheiri:2014lwa}.
Great progress has also been made in the realization of QEC by
virtue of tensor networks
\cite{Pastawski:2015qua,Freivogel:2016zsb,Mintun:2015qda,Bhattacharyya:2016hbx,Hayden:2016cfa,Qi:2018shh}.
In this framework, sub system duality is reflected by the
isometry between two sub Hilbert spaces associated with sub tensor
networks.

The above properties of entanglement have been addressed in
various tensor networks, including the multiscale entanglement
renormalization ansatz (MERA) \cite{Swingle:2009bg,Vidal:2007hda},
perfect tensor networks \cite{Pastawski:2015qua,Yang:2015uoa},
random tensor networks
\cite{Hayden:2016cfa,Han:2017uco,Bhattacharyya:2016hbx,Qi:2018shh,Chirco:2017wgl},
 hyperinvariant tensor
networks \cite{Evenbly:2017htn}, as well as spin networks
\cite{Han:2016xmb}.

Currently it is still a key issue whether tensor networks, or what
kind of tensor networks could produce all the aspects of
holography in the context of AdS/CFT correspondence. Taking
$AdS_3$/$CFT_2$ as an example, we pick up some important
properties that a tensor network is desired to possess.
\begin{itemize}
    \item Such a tensor network is a discretization of
    2-dimensional hyperbolic space ($H^2$ space), which is a time
    slice of an $AdS_3$ spacetime in global coordinate system.
    Correspondingly, the tensor network is endowed with a symmetry
    described by a discrete subgroup of $SL(2,R)$, which is the
    isometry of $H^2$ space.
    \item Such a tensor network respects RT formula and the entanglement entropy is characterized by a logarithmic
    law. Moreover, the entanglement spectrum (ES) of the ground state
    should be non-flat such that one can reproduce the Cardy-Calabrese
    formula of Renyi entropy for a $CFT_2$ with large central charge
    $c$, namely \cite{Calabrese:2004eu,Calabrese:2009qy}
    \be\label{CardyFormula}
    S_n(A)=\kc{1+\frac1n}\frac c6\ln l_A,
    \ee where $A$ is a spatial interval on the boundary and $l_A$ is
    its length with the unit of UV cutoff.
    \item Such a tensor network has the function of QEC as AdS spacetime enjoys.
    \item Such a tensor network can reproduce the behavior of Green's function in
    AdS$_3$/CFT$_2$.
\end{itemize}
Of course, all these properties may not be independent of one
another.

One candidate for capturing above holographic features of AdS
    is hyperinvariant tensor networks, recently proposed by Evenbly in
    \cite{Evenbly:2017htn}. It is composed of identical polygons by
    uniformly tiling hyperbolic space. The key idea is to impose
    constraints on the product of multiple tensors to form isometric
    mappings. It turns out that this sort of networks may combine the
    advantages of multiscale entanglement renormalization ansatz
    (MERA)
    \cite{Vidal:2007hda,Swingle:2009bg,Swingle:2012wq,Kim:2016wby,Bao:2015uaa}
    which is characterized by non-flat ES and the network composed of
    perfect
    tensors\cite{Pastawski:2015qua,Yang:2015uoa,Donnelly:2016qqt}
    which is usually endowed with the function of QEC.

But one key issue arises in this approach. That is, what kind
    of multi-tensor constraints could endow such features to a given
    tensor networks? or more quantitatively, is there any criteria to
    justify the ability of QEC and the non-flatness of ES for a given
    tensor networks with multi-tensor constraints? In
    \cite{Ling:2018qec} we have provided affirmative answers to these
    issues with the proposal for critical protection on tensor chains.
    In this paper we intend to elaborate our proposal and present the
    detailed analysis on tensor chains and constraints in tensor
    networks and prove the statements on the classification of tensor
    networks in \cite{Ling:2018qec}.

We organize the paper as follows. In next section we will propose
a generalized framework for the tensor networks with multi-tensor
constraints in the tiling of $H^2$ space. To classify different
types of multi-tensor constraints efficiently and describe the
behavior of tensor contractions during the evaluation of ES, we
introduce the notion of tensor chain to describe the contraction
of tensor products. Moreover, we will introduce a quantity, called
the average reduced interior angle, to characterize the geometric
structure of CP chain. Based on this structure we will introduce
the concept of critical protection in Section III, which should be
viewed as the core concept in our paper, because it plays an
essential role in measuring the quality of QEC as well as the
non-flatness of ES in a quantitative manner. As the first
consequence, we will immediately see that once the ES becomes
non-flat under the multi-tensor constraints as proposed in
\cite{Evenbly:2017htn}, then the ability of QEC from bulk to
boundary has to be weakened. Among this sort of networks, we find
that most of the perfect tensor networks as the limit case have
the strongest ability of QEC, while they are always accompanied by
a flat ES. Therefore, in order to construct tensor networks with a
non-flat ES as AdS spacetime, one has to pay the price of
sacrificing the ability of QEC. All above investigation is based
on a tensor networks embedded into $H^2$ space which can be
viewed as a discretization of the hyperbolic geometry.
Correspondingly, we may also describe QEC and ES over the geometry
of $H^2$ directly, which involves the notion of geodesics and the
curves of constant curvature, etc. We present the description
based on $H^2$ geometry in Section IV and the relevant backgrounds
are given in Appendix \ref{SectionH2}. Keep going on, to
intuitively understand the role of critical protection in the
evaluation of QEC and ES, in Section V we present some specific
examples of tensor networks and demonstrate how the realization of
QEC could be reflected by the structure of CP tensor chain, and
how the flatness of ES can be reflected by the region of critical
protection. Moreover, we develop a generalized description of
greedy algorithm by imposing multi-tensor constraints on tensor
chains. After that we classify tensor networks with constraints by
their properties of QEC and ES. We firstly study the relation
between CP and QEC in Section VI, presenting a criteria for the
existence of QEC, and then focus on the relation between CP and ES
in Section VII, with detailed proofs of the propositions on
various bounds for the flatness of ES. Section VIII is the
conclusion and outlook.

\section{Tensor chains in a tensor network}

In this section we will present a general framework for tensor
networks based on the tiling of hyperbolic space. We define a
notion of tensor chain whose skeleton forms a polyline in a
network. Associated with each tensor chain, the reduced interior
angle can be defined, which in some sense could be viewed as the
discrete description of the curvature of the ployline.

\subsection{Tiling of $H^2$ space}
In the global coordinate system of $AdS_3$ spacetime, the
isochronous surface is a $H^2$ space \be\label{H2Metric}
ds^2=L^2(d\rho^2+\cosh^2\rho d\tau^2), \ee where $L$ is the radius
of $H^2$ geometry, the unique dimensional quantity introduced
in this paper. So we are free to set $L=1$.

Firstly, we intend to discretize $H^2$ space in a uniform
version, which can be realized by the tiling of $H^2$ space with
identical polygons. Consider many identical polygons composed of
$b$ edges in a $2$ dimensional surface, then put them together by
gluing their edges such that $a$ edges share the same node. We
call such discretization as the $\{b,a\}$ tiling of $H^2$ space.
In a space with
negative curvature, because the sum of interior angles of a triangle is less than $2\pi$, one can realize a $\{b,a\}$ tiling of
$H^2$ space only if
\be\label{abconstraint}
\frac1a+\frac1b<\frac12.
\ee
Obviously, $a\geq3,b\geq3$.

When a tiling of $H^2$ is specified by $\{b,a\}$, the geometry is
determined (up to the radius $L$). We call the polygon with $b$
edges as the elementary polygon, while the union of several
elementary polygons as a composite polygon. The length of each
edge of the elementary polygon is \be\label{lengthp} P=2\text{
arccosh}\left(\frac{\cos \left(\frac{\pi }{b}\right)}{\sin
\left(\frac{\pi }{a}\right)}\right). \ee

\subsection{Tensor networks with $\{b,a\}$ tiling}
Now we construct a tensor network based on a $\{b,a\}$ tiling.
Associated with each node, we assign a rank-a tensor $T$, each
index of which is specified to an edge jointed at the node
respectively. The elements of the tensor $T$ are denoted as
$T^{i_1i_2\cdots i_a}$, where all indexes have the same dimension
$d$ and $d>1$. Associated with each edge, we also assign a rank-2
tensor $E$, whose elements are $E_{i_1i_2}$. We call the above
indexes associated with tensors $T$ and $E$ as basic indexes,
which are labelled by lowercase letters. As examples, two tensor
networks with $\ke{7,3}$ tiling and $\ke{4,5}$ tiling are
illustrated in Fig.\ref{Fig73Tiling}, respectively.

Because of the rotational invariance of $H^2$ space, we further
demand that the indexes of tensor $T$ and $E$ have cyclic
symmetry \footnote{Notice that the perfect tensor originally
    defined in \cite{Pastawski:2015qua} has no rotation symmetry as
    required. Here we further require it and will show the existence
    of such states in Appendix \ref{SectionExistence}.}
\be T^{i_1i_2\cdots i_a}=T^{i_2i_3\cdots i_ai_1},\quad
E_{i_1i_2}=E_{i_2i_1}. \ee In this paper, we adopt the convention
that the index of tensor $T$ can be lowered by contracting it
with a tensor $E$ \be T_{j_1}{}^{i_2\cdots i_a}\equiv
\sum_{i_1}E_{j_1i_1}T^{i_1i_2\cdots i_a}. \ee Correspondingly, the
edge connecting two nodes represents the index contraction of two
tensors $T$ by a tensor $E$, namely, \be
\sum_{i_1j_1}T^{i_1i_2\cdots i_a}E_{i_1j_1}T^{j_1j_2\cdots
j_a}=\sum_{i_1}T^{i_1i_2\cdots i_a}T_{i_1}{}^{j_2\cdots
    j_a}.
\ee Therefore, given a tensor network with $\ke{b,a}$ tiling, we
can define a quantum state $\Psi$ consisting of two sorts of
tensors $T$ and $E$ by tensor products and contractions. For later
convenience, we require that for a tensor network state $\Psi$,
all the indexes of tensors $T$ (with full upper indexes) should be
contracted and those uncontracted indexes should only belong to
tensors $E$, as shown in Fig.\ref{Fig73Tiling}.

A tensor network $\Psi$ can define a state $\ket{\Psi}$ in the
Hilbert space on those uncontracted edges. In this paper, we will
investigate the algorithms of QEC and the entanglement of $\Psi$
by manipulating tensor networks.

\begin{figure}
    \centering
    \subfigure[]{
    \includegraphics[width=0.4\linewidth]{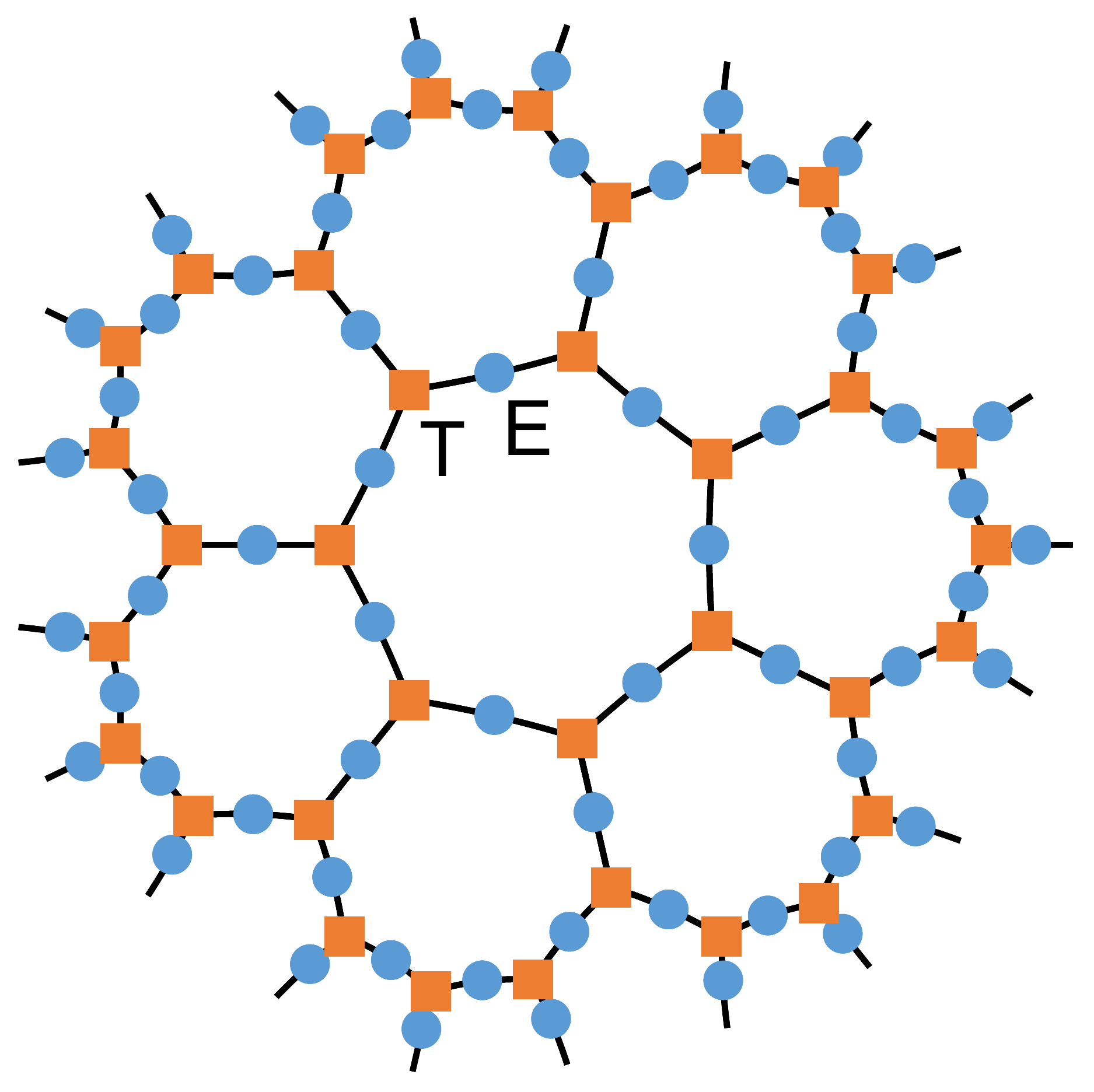}
    \label{Fig73Tiling}}
    \subfigure[]{
    \includegraphics[width=0.4\linewidth]{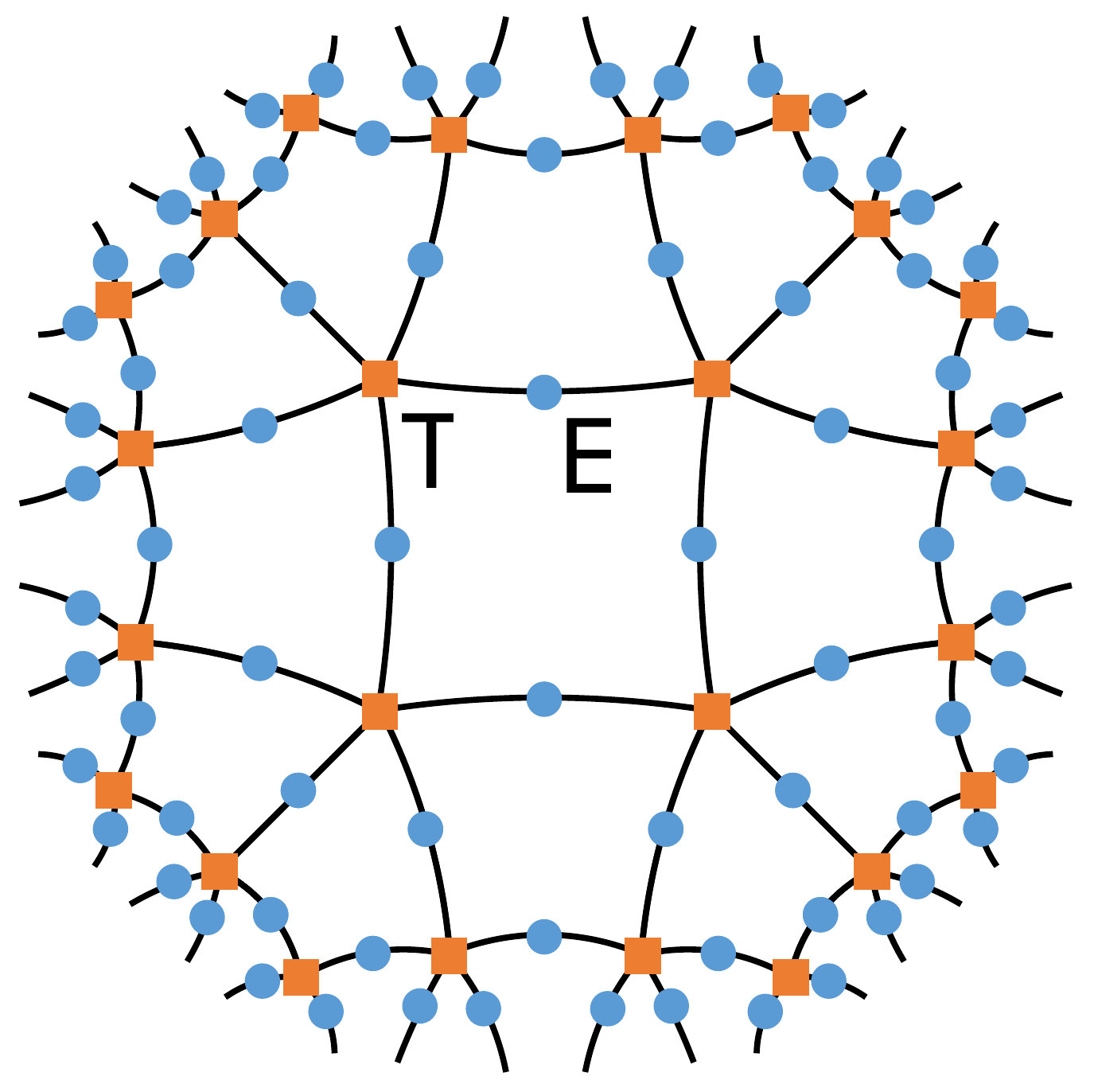}
    \label{Fig45Tiling}
    }
\caption{Two tensor networks with (a) $\ke{7,3}$ tiling and (b)
$\ke{4,5}$ tiling, where a tensor $T$ is marked by an orange
square while a tensor $E$ is marked by a blue dot. Only finite
layers are plotted. For convenience, all the tensors $T$ and
$E$ will not be manifestly displayed in the latter figures of
tensor network except necessary.}
\end{figure}

\subsection{Tensor chain}
Given a tensor network by the $\{b,a\}$ tiling, we intend to
introduce a notion of tensor chain to depict the product structure
of multi-tensors with index contractions, which will be convenient
for us to impose a tensor constraint and quantitatively describe
its geometric properties. Firstly, in order to define a tensor
chain in an efficient way, we adopt a compact form to denote a
single tensor $T$ of rank-$a$ which is subject to rotational
symmetry. We divide all its indexes into four groups {\it in
order} and label each group with an abstract index, which is
called collected index and labelled by a capital letter. Then
the component of $T$ can be generally written as $T^{ABCD}$. For
instance, for a tensor $T^{i_1i_2\cdots i_5}$ with $5$ indexes, we
may collect them as $i_1i_2=A,\,i_3=B,\,i_4=C,\,i_5=D$ or
$i_1i_2=A,\,i_3i_4=B,\,i_5=C,\,\varnothing=D$, while the
collection $i_1i_2=A,\,i_3=C,\,i_4=B,\,i_5=D$ is prohibited.
Furthermore, those collected indexes may also be lowered by tensor
$E$, such as \be T_{A'}{}^{BCD}\equiv
T_{i_1'i_2'}{}^{i_3i_4i_5}\equiv
\sum_{i_1i_2}E_{i_1'i_1}E_{i_2'i_2}T^{i_1i_2\cdots i_5}\equiv
\sum_{A}E_{A'A}T^{ABCD}, \ee where
$i_1'i_2'=A',\,i_1i_2=A,\,i_3=B,\,i_4=C,\,i_5=D$ and
$E_{A'A}\equiv E_{i_1'i_1}E_{i_2'i_2}$. We further define $\#(A)$
as the number of basic indexes in a collected index $A$. For
example, in the above collection, $\#(A)=2$ and $\#(B)=1$.

Now we can construct a tensor chain $M$ by contracting $k$
tensors $T$ with tensors $E$
\be\label{TensorChain}
M^A_B=M^{A_1A_2\cdots A_k}_{B_1B_2\cdots B_k}
=\sum_{C_1C_2\cdots C_{k},C_1=C_{k+1}}\prod_{i=1}^k T^{A_iC_i}{}_{B_iC_{i+1}}
\ee where $A=A_1A_2\cdots A_k$ and $B=B_1B_2\cdots B_k$ are
uncontracted indexes, while $C_1,C_2,\cdots C_k$ are specified as
the indexes which are contracted in the chain. Moreover, since
each tensor $T$ occupies a node in the tiling, we call $k$ the
number of nodes in tensor chain $M$ as well. The sum index $i$ in
(\ref{TensorChain}) runs from $1$ to $k$, counting the number of nodes in $M$. Alternatively, a tensor chain $M$ can be viewed as a mapping from the Hilbert space on uncontracted indexes $A$ to the Hilbert space on uncontracted indexes $B$.

Obviously, in a $\{b,a\}$ tiling any two nodes are only possibly connected by single edge, otherwise they are not connected directly. Thus, here we only consider tensor chain with $\#(C_i)=1$ for $i=2,3,\cdots,k$.
Furthermore, if $\#(C_1)=\#(C_{k+1})=1$, then we call $M$ as a
closed tensor chain; if $\#(C_1)=\#(C_k)=0$, we call $M$ as an
open tensor chain. Two typical samples of tensor chain are
illustrated in Fig.\ref{ChainMatrix}.

Since a diagram of tensor chain can be specified by the number of
uncontracted edges in the tensor product, we propose a notation
$\bbm m_1&m_2&\cdots&m_k\\ n_1&n_2&\cdots&n_k \ebm$ to denote an
open tensor chain $M$ with elements $M^{A_1A_2\cdots
    A_k}_{B_1B_2\cdots B_k}$, where $m_i\equiv\#(A_i)$ and
$n_i\equiv\#(B_i)$. Obviously, we have
\be\label{Opensummn}
m_i,n_i\geq0, \quad m_i+n_i=a-2+\delta_{i1}+\delta_{i,k},\quad i=1,2,\cdots,k.
\ee Similarly, we use $\bpm m_1&m_2&\cdots&m_k\\
n_1&n_2&\cdots&n_k \epm$ to denote a closed tensor chain with
\be\label{Closedsummn} m_i,n_i\geq0, \quad m_i+n_i=a-2,\quad
i=1,2,\cdots,k.
\ee
Since one can reconstruct $m_i$ from $n_i$
or vice versa according to (\ref{Opensummn})(\ref{Closedsummn}),
some time for convenience we abbreviate
either $m_i$ or $n_i$ to $*$. For instance, $\bbm *&*&\cdots&*\\
n_1&n_2&\cdots&n_k \ebm \equiv \bbm m_1&m_2&\cdots&m_k\\
*&*&\cdots&* \ebm \equiv \bbm m_1&m_2&\cdots&m_k\\
n_1&n_2&\cdots&n_k \ebm$ with (\ref{Opensummn}) and $\bpm
*&*&\cdots&*\\ n_1&n_2&\cdots&n_k \epm \equiv \bpm
m_1&m_2&\cdots&m_k\\ *&*&\cdots&* \epm \equiv \bpm
m_1&m_2&\cdots&m_k\\ n_1&n_2&\cdots&n_k \epm$ with
(\ref{Closedsummn}).

We provide the following definition to describe some
relations among tensor chains.
\begin{definition}
    Given a tensor network with $\{b,a\}$ tiling and an open tensor chain $M=\bbm m_1&m_2&\cdots&m_k\\ *&*&\cdots&* \ebm$, we define its sub tensor chain $M_{p,q}$ as $\bbm m_p&m_{p+1}&\cdots&m_{q}\\ *&*&\cdots&* \ebm$ for $1\leq p\leq q\leq k$. We say $M'\sqsubseteq M$ if $M'$ is a sub tensor chain of $M$. Furthermore, we say an open tensor chain $\bbm m_1'&m_2'&\cdots&m_l'\\ *&*&\cdots&*\ebm\preceq  M$ where $1\leq l\leq k$, if $\exists p\in \ke{1,2,\cdots,k-l+1}$ {\it s.t.} $m_i' \leq m_{p+i-1}$ for $i\in \ke{1,l}$ and $m_i'= m_{p+i-1}$ for $1<i<l$.
    
    Given a closed tensor chain $M=\bpm m_1&m_2&\cdots&m_k\\ *&*&\cdots&* \epm$, we define its sub tensor chain $M_{p,q}$ for $1\leq p\leq k$, $1\leq q\leq k$ as $\bbm m_1''&m_2''&\cdots&m_l''\\ *&*&\cdots&*  \ebm$ where $m_i''=m_{p+i-1 \text{ mod } k}$ for $1\leq i\leq l$ and $1\leq l=q-p+1<k$. We say $M'\sqsubseteq M$ if $M'$ is a sub tensor chain of $M$. Furthermore, we say an open tensor chain  $\bbm m_1'&m_2'&\cdots&m_l'\\ *&*&\cdots&*  \ebm \preceq M$ where $1\leq l<k$ if $\exists p$ {\it s.t.} $m_i'\leq m_{p+i-1 \text{ mod } k}$ for $i=\ke{1,l}$ and $m_i'=m_{p+i-1 \text{ mod } k}$ for $1<i<l$.
\end{definition}
For example, in $\ke{4,5}$ tiling, $\bbm 2&1\\2&3 \ebm\sqsubseteq \bbm 2&1&2\\2&2&2 \ebm$ and $\bbm 2&1&1\\2&2&3 \ebm\preceq \bbm 2&1&2\\2&2&2\ebm$.

We always require that the subscripts $p,q$ in $M_{p,q}$ belong to
integers and satisfy $1\leq p\leq q\leq k$, where $k$ is the
number of nodes in $M$.

\begin{figure}
    \centering
    \includegraphics[width=150pt]{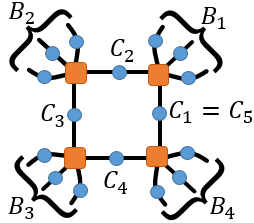}\quad\quad\quad\quad
    \includegraphics[width=150pt]{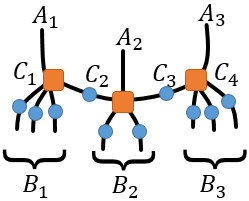}
    \caption{A closed tensor chain $\bpm 0&0&0&0\\ 3&3&3&3 \epm$ and an open tensor chain $\bbm 1&1&1 \\ 3&2&3 \ebm$.}\label{ChainMatrix}
\end{figure}

\begin{figure}
    \centering
    \includegraphics[width=150pt]{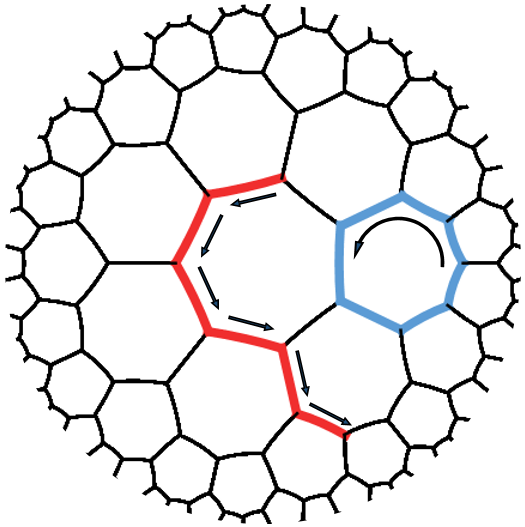}
    \includegraphics[width=150pt]{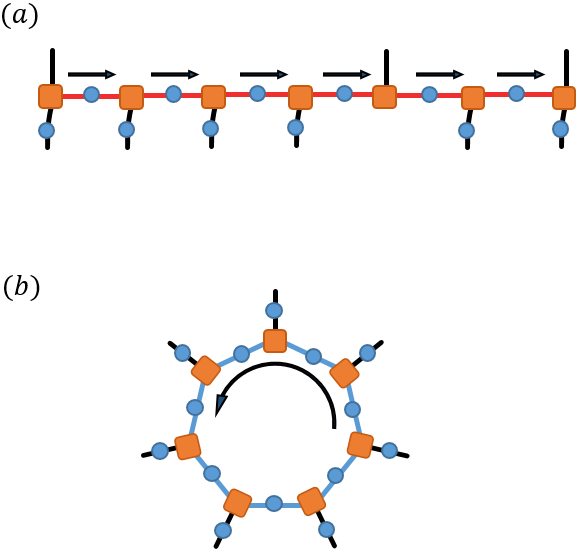}
    \caption{Two polylines and their tensor chains in a tensor network with $\ke{7,3}$ tiling. The tensor chain of red polyline is shown in (a) and the tensor chain of blue polyline is shown in (b). }\label{FigPolyline}
\end{figure}

\subsection{The reduced interior angle}

When a tensor chain is embedded into a tensor network in $H^2$
space, its skeleton can be marked by a directed polyline
concisely, as shown in Fig.\ref{FigPolyline}. Along the direction
of the polyline, we require that the sequence number of nodes
increases and the edges on the left (right) hand side of the
polyline are always associated with the upper (lower) indexes
of the tensor chain. For a closed tensor chain, conventionally
the direction of the closed polyline is specified to be
anticlockwise, so the inward or left-handed (outward or
right-handed) edges of the polyline are associated with the upper
(lower) indexes of the tensor chain. We remark that a closed
polyline with clockwise direction can be analyzed in parallel,
with the requirement that its inward (outward) edges are
associated with the lower (upper) indexes.

The curvature of a polyline at the $i$th node can be captured by
its interior angle $\theta_i$, which is defined as the angle on
the left hand side of the polyline and is a multiple of $2\pi/a$.
We further define the reduced interior angle as $s_i=\theta_i/\frac{2\pi}{a}$, which is an integer. Obviously,
the reduced interior angle is related to the number of upper
edges at each node and we intend to give the following definition.
\begin{definition}\label{definitions}
    For a closed tensor chain $M=\bpm
    m_1&m_2&\cdots&m_k\\ *&*&\cdots&* \epm$, the reduced interior
    angle of the i-th tensor is $s_i=m_i+1$. For an open tensor chain
    $M=\bbm m_1&m_2&\cdots&m_k\\ *&*&\cdots&* \ebm$, the reduced
    interior angles are $s_i=m_i+1-\delta_{i1}$.
\end{definition}
For later convenience, we further introduce several quantities
based on reduced interior angles to evaluate the curvature of
a tensor chain $M$. Specially, we define the prime tensor chain,
which is the core notion for the construction of the algebra of
tensor constraints in next subsection.
\begin{definition}
    Given a tensor chain $M$ with $k$ nodes, the average reduced interior angle $\kappa(M)$ is defined as
    \be\label{kappa}
    \kappa(M)=\frac{1}{k}\sum_{i=1}^{k} s_i.
    \ee
    The sub reduced interior angle $\kappa_{p,q}(M)$ from the
    p-th tensor to the q-th tensor is defined as
    \be
    \kappa_{p,q}(M)=\kappa(M_{p,q}), \quad 1\leq p\leq q\leq k,
    \ee
    and the maximal reduced interior angle $\kappa_{p,q}(M)$
    is defined as
    \be
    \kappa_{max}(M)=\max_{p,q} \kappa_{p,q}(M).
    \ee
\end{definition}

\begin{definition}
    Given an open tensor chain $M$ with finite $k$ nodes, when
    that $\kappa_{p,q}(M)=\kappa_{max}(M)$ if and only if $p=1$ and
    $q=k$, then we say tensor chain $M$ is prime.
\end{definition}

Based on above definitions, we have following theorems for
tensor chains.

\begin{theorem}\label{theoremkappamaxpreceq}
    Given two tensor chains $M$ and $M'$, if $M\preceq M'$, then
    $\kappa_{max}(M)\leq\kappa_{max}(M')$.
\end{theorem}

\begin{theorem}\label{theorempreceqprime}
    Given a tensor chains $M$ and a prime tensor chain $M'$, if $M\preceq M'$ but $M\neq M'$, then
    $\kappa_{max}(M)<\kappa(M')$.
\end{theorem}

\begin{theorem}\label{theoremprime}
    {\bf The uniqueness theorem of prime tensor chain.} Given a
    $\{b,a\}$ tiling and a rational number  \be\label{kappaRange}
    \kappa\in \kd{1,\frac a2},
    \ee the prime tensor chain $M$ with $\kappa(M)=\kappa$ exists and
    is unique. Furthermore, if $\kappa=u/v$ where $u$ and $v$ are
    coprime integers, then $M=\bbm
    m_1&m_2&\cdots&m_v\\n_1&n_2&\cdots&n_v \ebm$ with
    \be\label{constructm}
    m_i=[i\kappa]-[(i-1)\kappa]-1+\delta_{i1},
    \ee
    where $[x]$ is the floor function. So, $M$ has reversal
    symmetry, namely \be\label{reversalsymmetry}
    m_i=m_{v-i+1}.
    \ee The reduced interior angles of $M$ are given by
    \be\label{constructs}
    s_i=[i\kappa]-[(i-1)\kappa].
    \ee
\end{theorem}

\begin{proof}
    Set a prime tensor chain $M=\bbm m_1&m_2&\cdots&m_k\\n_1&n_2&\cdots&n_k \ebm$ satisfying $\kappa(M)=\kappa$.
    Then, for $l=1,2,\cdots,k-1$,
    \bea
    \frac1k\sum_{i=1}^k m_i+k-1&=&\kappa,  \label{kappak}\\
    \frac1l\sum_{i=1}^l m_i+l-1&<&\kappa,  \label{kappal}\\
    \frac1{k-l}\sum_{i=1}^{k-l} m_i+k-l-1&<&\kappa.  \label{kappakl}
    \eea
    On one hand, from (\ref{kappak})(\ref{kappal})(\ref{kappakl}),
    one has
    \be\label{kappalk}
    \sum_{i=1}^l m_i-1<l(\kappa-1)<\sum_{i=1}^l m_i.
    \ee
    Then
    \be
    l(\kappa-1)\notin\mathbb{Z}\Rightarrow l/v\notin\mathbb{Z},\quad
    l=1,2,\cdots,k-1.
    \ee
    Thus
    \be
    v\geq k.
    \ee
    On the other hand, from (\ref{kappal}), one has
    \be
    k\kappa\in\mathbb{Z}\Rightarrow k/v\in\mathbb{Z}.
    \ee
    Combining above two statements, we have
    \be
    k=v.
    \ee
    From (\ref{kappal})(\ref{kappalk}), one has
    \be
    \sum_{i=1}^l m_i=[l(\kappa-1)]+1, \quad  l=1,2,\cdots,v.
    \ee
    So
    \be
    m_l=\sum_{i=1}^l m_i-\sum_{i=1}^{l-1} m_i=[l\kappa]-[(l-1)\kappa]-1, \quad
    l=2,3,\cdots,v.
    \ee Specially, when $l=1$, \be
    m_1=[1*\kappa]-[0*\kappa]-1+1.
    \ee In summary, we have proved (\ref{constructm}). Moreover,
    using identities
    \be\label{trickfloor}\begin{split}
        [-x]&=-[x]-1+\delta_{x,[x]},    \\
        [x+n]&=[x]+n,\quad n\in \mathbb{Z},
    \end{split}\ee
    we can derive (\ref{reversalsymmetry}). Then according to
    Definition \ref{definitions}, we have (\ref{constructs}). At
    last, from (\ref{constructs}) we find \be
    m_{1}=[\kappa] \text{ or } [\kappa]-1.
    \ee
    It leads to the conclusion that conditions in
    (\ref{Opensummn}) can always be satisfied if we require
    (\ref{kappaRange}). So the prime tensor chain exists and is
    unique.
\end{proof}

\begin{theorem}\label{theoremcentralM}
    For a tensor chain $M$, $\exists$ unique prime tensor chain $M'$ satisfying $M'\sqsubseteq M$ and $\kappa(M')=\kappa_{max}(M)$.
\end{theorem}
\begin{proof}
    $\exists p,q$ {\it s.t.} $\kappa_{p,q}(M)=\kappa_{max}(M)$ and $\kappa_{p',q'}(M)<\kappa_{max}(M)$, $\forall p',q'$ satisfying $p< p'\leq q'\leq q$ or $p\leq p'\leq q'< q$. Then $M'=M_{p,q}$ is prime, whose $\kappa(M')=\kappa_{p,q}(M)=\kappa_{max}(M)$. Because of Theorem \ref{theoremprime}, $M'$ is unique.
\end{proof}

\section{Tensor constraints and Critically Protected (CP) tensor chains}\label{SectionCP}

In this section we propose a notion of critical protection to
describe the behavior of tensor networks under the
contractions of tensor product which are subject to tensor
constraints.

\begin{figure}
    \centering
    \includegraphics[width=200pt]{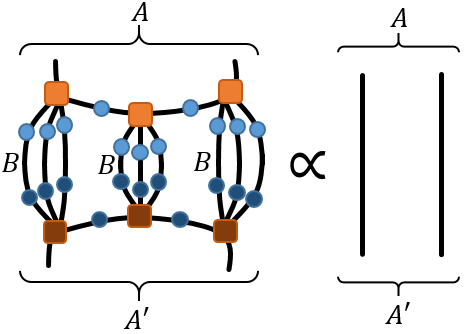}\\
    \caption{Multi-tensor constraint constructed by open tensor chain $\bbm 1&0&1\\3&3&3\ebm$.}\label{Figisoeq}
  \end{figure}

\subsection{Tensor constraint}\label{SubSectionTensorConstraint}
The notion of tensor chain provides us a convenient way to
describe a general constraint on the product of tensors $T$ and
$E$, which plays an essential role in pushing operators through
nodes or edges in network in the context of QEC. Usually we impose
the constraint requiring that some contraction of tensors
should be proportional to an isometry. Of course the
contraction of tensors may or may not form a tensor chain. Here
for simplicity, we only consider imposing tensor constraints on
tensor chains which can be concisely written as \be\label{IsoE}
    \sum_{B}  M^A_B (M_B^{A'})^*   \propto\delta^{AA'},
\ee where $M=\bbm m_1&m_2&\cdots&m_k\\n_1&n_2&\cdots&n_k \ebm$ is
an open tensor chain with $n_i$ to be the number of edges that are
contracted with its conjugate tensor at the $i$th node, as
illustrated in Fig.\ref{Figisoeq}. Notice that the contraction on
two lower indexes $B$ in (\ref{IsoE}) implies that it involves in
$\#(B)$ contractions of $\sum_j E_{ij}(E_{jk})^*$. For
convenience, in the remainder of this paper, when we say a tensor
constraint $M$, we actually refer to the constraint in terms of
tensor chain $M$ which is subject to (\ref{IsoE}). Obviously, a
non-trivial constraint $M$ requires $\sum_{i=1}^k m_i\geq1$.
Moreover, an isometry can be realized only if the number of
degrees of freedom in $A$ is less than or equal to that in $B$,
namely \be\label{freedomless}
    \sum_{i=1}^k m_i=\#(A)\leq\#(B)=\sum_{i=1}^k n_i.
\ee
Thus, any non-trivial tensor constraint $M$ should satisfy
\be\label{kappaMRange}
    \kappa(M)\in \kd{1,\frac a2}.
\ee

One may immediately find that for a set of tensor constraints some
of them may not be logically independent. In general, there are
four fundamental operations to derive new constraints from given
tensor constraints, which can be listed as follows.
\begin{description}
  \item[Reversal] If $\bbm m_1&m_2&\cdots&m_k\\n_1&n_2&\cdots&n_k \ebm$ is a tensor constraint, then $\bbm m_k&m_{k-1}&\cdots&m_1\\n_k&n_{k-1}&\cdots&n_1 \ebm$ is a tensor constraint as well.
  \item[Contraction] If $\bbm m_1&m_2&\cdots&m_k\\n_1&n_2&\cdots&n_k \ebm$ is a tensor constraint, then $\bbm m_1-1&m_2&\cdots&m_k\\n_1+1&n_2&\cdots&n_k \ebm$ is a tensor constraint as well.
  \item[Reduction] If $\bbm m_1&m_2&\cdots&m_k\\n_1&n_2&\cdots&n_k \ebm$ and $\bbm *&*&\cdots&*\\n_{p}&n_{p+1}&\cdots&n_q \ebm$ are tensor constraints, where $2\leq p\leq q\leq k$, then $\bbm m_1&m_2&\cdots&m_{p-1}\\ *&*&\cdots&* \ebm$ is a tensor constraint as well.
  \item[Combination] If $\bbm *&*&\cdots&*\\n_1&n_2&\cdots&n_k \ebm$ and $\bbm m_1'&m_2'&\cdots&m_l'\\ *&*&\cdots&* \ebm$ are tensor constraints, then \\ $\bbm *&*&\cdots&*&m_1'&m_2'&\cdots&m_l'\\n_1&n_2&\cdots&n_k&*&*&\cdots&* \ebm$ is a tensor constraint as well.
\end{description}
As an example, we demonstrate the derivation of new constraints
by contraction and reduction in Fig.\ref{Figlogicallyderivation}.

\begin{figure}
    \centering
    \includegraphics[height=200pt]{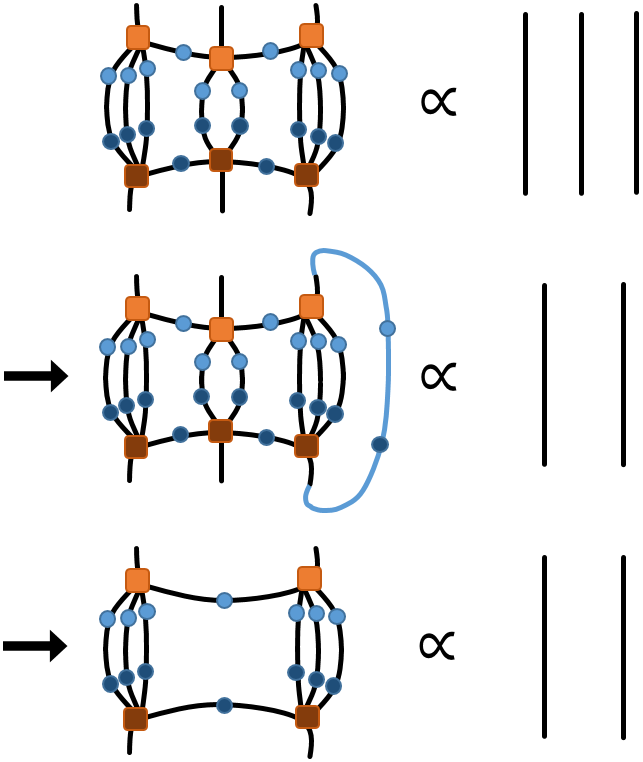}\\
    \caption{Logical derivation of tensor constraint $\bbm 1&1\\3&3\ebm$ from $\left\{\bbm 1\\4\ebm, \bbm 1&1&1\\3&2&3\ebm\right\}$.}\label{Figlogicallyderivation}
  \end{figure}

We remark that the strength of a tensor constraint $M$ can be
quantified by its maximal reduced interior angle
$\kappa_{max}(M)$. By comparing $\kappa_{max}$ of new derived
constraints with those of original constraints, we have the
following theorem.
\begin{theorem}\label{theoremkappamaxnotincrease}

Any tensor constraint $M''$ derived from $M$ and $M'$ through
above ways satisfies
    \be\label{kappamaxnotincrease}
    \kappa_{max}(M'')\leq \max\kc{\kappa_{max}\kc{M},\kappa_{max}\kc{M'}}.
    \ee
where $M=\bbm m_1&m_2&\cdots&m_k\\n_1&n_2&\cdots&n_k \ebm,M'=\bbm m_1'&m_2'&\cdots&m_l'\\n_1'&n_2'&\cdots&n_l' \ebm$.
\end{theorem}
\begin{proof}
  For reversal, $\kappa_{max}(M'')=\kappa_{max}(M)$, then we have (\ref{kappamaxnotincrease}).    \\
  For contraction and reduction, $M''\preceq M$. Thanks to Theorem \ref{theoremkappamaxpreceq}, we have (\ref{kappamaxnotincrease}). \\
  For combination, $\exists p,q$ {\it s.t.} $\kappa_{p,q}(M'')=\kappa_{max}(M'')$. In the case of $q\leq k$ or $k<p$, similarly, we have (\ref{kappamaxnotincrease}).
  In the case of $p\leq k<q$, we observe that
  \bea
    1+\frac{\sum_{i=p}^k m_i-1}{k-p+1}&\leq& \kappa_{max}(M),  \\
    1+\frac{\sum_{i=1}^{k-q} m_i'-1}{q-k}&\leq& \kappa_{max}(M'),
  \eea
which leads to
  \be
    \sum_{i=p}^k m_i+\sum_{i=1}^{q-k} m_i'\leq 2+(\kappa_{max}(M)-1)(k-p+1)+(\kappa_{max}(M')-1)(q-k).
  \ee
  So,
  \be\begin{split}
    \kappa_{p,q}(M'')
    &=1+\frac{\sum_{i=p}^k m_i+\sum_{i=1}^{q-k} m_i'-2}{q-p+1}  \\
    &\leq \frac{k-p+1}{q-p+1}\kappa_{max}(M)+\frac{q-k}{q-p+1}\kappa_{max}(M')    \\
    &\leq \max\kc{\kappa_{max}\kc{M},\kappa_{max}\kc{M'}}.
  \end{split}\ee
In summary, we have (\ref{kappamaxnotincrease}).
\end{proof}

Next we intend to study a set of tensor constraints with the form
\be\label{GeneralSet} S=\left\{M_s, M_1,M_2,\cdots\right\}, \ee
where $M_s=\bbm 1\\a-1\ebm$ is called step tensor chain and other
tensor constraints are not specified. Step tensor chain has the
minimal average reduce interior angle $\kappa(M_s)=1$. Generally,
those tensor constraints in $S$ may not be mutually independent.
With the help of step tensor chain, we have the following theorem
for the relation of tensor constraints.
\begin{theorem}\label{theorempreceq}
    Any tensor constraint $M$ satisfying $M\preceq M'$ can be derived
    from $S=\left\{M_s, M'\right\}$.
\end{theorem}

For any general set of tensor constraints, we can find a unique
set of tensor constraints $S_c$ which is logically equivalent to
$S$ and only contains two elements \be\label{CentralSet}
    S_c=\left\{M_s, M_t\right\}.
\ee $S_c$ is called central set. $M_t$ is called top tensor chain,
which should be prime and satisfy the condition
(\ref{kappaMRange}). The equivalence between $S$ and $S_c$ require
$\kappa(M_t)=\max\limits_{M\in S} \kappa_{max}(M)$, as proved in
Theorem (\ref{theoremSet}). Without loss of generality, we will
only consider central set $S_c$ hereinafter.

Thanks to Theorem \ref{theoremprime}, given a $\ke{b,a}$ tiling,
we have a one-to-one mapping between all the possible top tensor
chains $M_t$ and the rational numbers in $\kd{1,\frac a2}$. Thus,
we have classified all the general sets of tensor constraints with
the form in (\ref{GeneralSet}) by the rational number
$\kappa(M_t)\in [1,\frac a2]$. Given a $\kappa(M_t)$, with the use
of (\ref{constructm}), we can directly construct top tensor chain
$M_t$ as well as those tensor chains $M$ satisfying $M\preceq
M_t$.

\begin{theorem}\label{theoremkappaleq}
    Given a central set $S_c=\left\{M_s,
    M_t\right\}$, we define a set $S_D=\ke{M|\kappa_{max}(M)\leq \kappa(M_t)}$.
 A tensor constraint $M$ can be derived from $S_c$ if and only if $M\in S_D$.
\end{theorem}

\begin{proof}
    Let $M_t=\bbm m_1&m_2&\cdots&m_k\\n_1&n_2&\cdots&n_k \ebm,
    M=\bbm m_1'&m_2'&\cdots&m_{k'}'\\n_1'&n_2'&\cdots&n_{k'}' \ebm$.
    We denote the proposition that $M$ can be derived from $S_c$ as $P1$.

    If  $P1$ is true, thanks to Theorem \ref{theorempreceqprime} and Theorem
    \ref{theoremkappamaxnotincrease}, we have $\kappa_{max}(M)\leq\kappa(M_t)$, thus $M\in S_D$.

Next we will apply the method of induction on $k'$ to prove that
if $\kappa_{max}(M)\leq\kappa(M_t)$ then  $P1$ is true.

Firstly, for the simplest case with $k'=1$, $\kappa_{max}(M)=m_1'$
is an integer. Because of Theorem \ref{theoremprime},
$m_1=[\kappa(M_t)]$. Then $\kappa_{max}(M)\leq
\kappa(M_t)\Rightarrow m_1'\leq m_1\Rightarrow M\preceq M_t$.
Because of Theorem \ref{theorempreceq}, $P1$ is true.

Now assume that when $1\leq k'<l$ $P1$ is true, we are going to
prove that for $k'=l$, $P1$ is also true.

At current stage, the length of tensor chain $M$, namely $l$,
could be either longer or shorter than the length of $M_t$, namely
$k$. In either case, we will compare the number of upper basic
indexes at each node within the parts with the same length,
$\min(k,l)$. To describe the difference of this part in two tensor
chains, it is convenient to define the proposition that $\exists
j\in\ke{1,2,\cdots,\min(k,l)}$ {\it s.t.} $m_j'\neq m_j$ as
 $P2$. Then for $k'=l$, We split the situation into the
following three cases.
    \begin{enumerate}
\item  $P2$ is false and $l\leq k$. It means that the tensor chain
$M$ is shorter or has equal length, and has the same number of
upper basic indexes as $M_t$ at each node. Then obviously one has
$M\sqsubseteq M_t$, so  $P1$ is true. \item $P2$ is false and
$l>k$. It means the tensor chain $M$ has the same number of upper
basic indexes as $M_t$ at first $k$ nodes but it is longer. One
can pick out the extra part of $M$
by setting $M'=\bbm m_{k+1}'+1&m_{k+2}'&\cdots&m_l'\\
*&*&\cdots&* \ebm$. For $p>1$, one can easily derive that
$\kappa_{p,q}(M')=\kappa_{p+k,q+k}(M)\leq
\kappa_{max}(M)\leq\kappa(M_t)$. While for $p=1$, one has
$\kappa_{1,q}(M')=\frac1q \ke{(k+q)\kappa_{1,k+q}(M)-k\kappa(M_t)}
\leq \frac1q \ke{(k+q)\kappa(M_t)-k\kappa(M_t)}=\kappa(M_t)$. So
$\kappa_{max}(M')\leq \kappa(M_t)$. Because of the assumption of
induction and $l-k<l$, $M'$ can be derived from $S_c$. Now since
$M$ can be derived from $M'$ and $M_t$ by reversal and
combination,  $P1$ must be true. \item  $P2$ is true. It means the
number of upper basic indexes at some nodes are different in two
tensor chains. Set the minimal $j$ satisfying $m_j'\neq m_j$ as
$r$. If $m_r'>m_r$, then $m_r'\geq m_r+1$. So $\kappa_{1,r}(M)
\geq \frac1r \kc{m_1+\cdots+m_r+1+r-1} =\frac1r
\ke{\sum_{i=1}^r([i\kappa(M_t)]-[(i-1)\kappa(M_t)]-1+\delta_{i1})+r}
=\frac1r \kc{[r\kappa(M_t)]+1} >\kappa(M_t)$, where
(\ref{constructm}) is used. It is contradictory to the condition
$\kappa_{max}(M)\leq\kappa(M_t)$. So only $m_r'<m_r$ is possible.
Now we set $M'=\bbm m_1'&m_2'&\cdots&m_{r-1}'&m_r'+1\\
*&*&\cdots&*&* \ebm$ and $M''=\bbm m_{r+1}'&m_{r+2}'&\cdots&m_l'\\
*&*&\cdots&* \ebm$ (If $r=1$, set $M'=\bbm m_1'+1\\ * \ebm$), then
$M'\preceq M_t$ and $M''\preceq M$. So $M'$ can be derived from
$S_c$ and $\kappa_{max}(M'')\leq \kappa(M)$. Because of the
assumption of induction and $l-r<l$, $M''$ can be derived from
$S_c$ as well. Since $M$ can be derived from $M'$ and $M''$ by
combination, $P1$ is true.
    \end{enumerate}
    In conclusion, $P1$ is true when $\kappa_{max}(M)\leq\kappa(M_t)$.
\end{proof}

\begin{theorem}\label{theoremSet}
Any general set of tensor constraints $S=\ke{M_s,M_1,M_2,\cdots}$
is logically equivalent to a unique central set of tensor
constraints $S_c=\ke{M_s,M_t}$, where $M_t$ is specified by
$\kappa(M_t)=\max\limits_{M\in S} \kappa_{max}(M)$.
\end{theorem}
\begin{proof}
    $\exists M\in S$ {\it s.t.} $\kappa_{max}(M)=\max\limits_{M'\in S} \kappa_{max}(M')$. Because of Theorem \ref{theoremcentralM}, $\exists$ unique prime tensor chain $M_t$ satisfying $M_t\sqsubseteq M$ and $\kappa(M_t)=\kappa_{max}(M)$. Because of Theorem \ref{theorempreceq}, $S\Rightarrow \ke{M_s,M} \Rightarrow S_c$ and $\kappa(M_t)=\max\limits_{M\in S} \kappa_{max}(M)$.

    On the other hand, because of Theorem \ref{theoremkappaleq}, $S_c\Rightarrow \ke{M_s,M}$, $\forall M\in S$. Thus, $S_c\Rightarrow S$.

    Thus, $\exists $ unique prime $ M_t$ {\it s.t.} $S \Leftrightarrow S_c$ and $\kappa(M_t)=\max\limits_{M\in S} \kappa_{max}(M)$.
\end{proof}

One may ask whether there always exist tensor $T$ and tensor
$E$ satisfying the tensor constraints $S_c$. We do not have a
general proof about the existence here. Nevertheless, for some
specific tensor constraints, we can actually solve them by
constructing specific tensors indeed. Some examples are
demonstrated in Appendix \ref{SectionExistence} and see more in \cite{Ling:2018qec}.

\subsection{Protection}

In this subsection we describe the behavior of tensor chains
under the action of tensor constraints. For this purpose we first
give the following definitions.
\begin{definition}
The transpose of an open tensor chain $M=\bbm m_1&m_2&\cdots&m_k\\
n_1&n_2&\cdots&n_k \ebm$ is $M^T=\bbm
n_1&n_2&\cdots&n_k\\m_1&m_2&\cdots&m_k \ebm$. The transpose of a
closed tensor chain $M=\bpm m_1&m_2&\cdots&m_k\\
n_1&n_2&\cdots&n_k \epm$ is $M^T=\bpm
n_1&n_2&\cdots&n_k\\m_1&m_2&\cdots&m_k \epm$.
\end{definition}

\begin{definition}\label{definitionProtect}
Given a set of tensor constraints $S_c=\left\{M_s, M_t\right\}$,
say a tensor chain $M$ is unprotected, if $\exists M'\in S_D$ such
that $M'^T\sqsubseteq M^T$. Otherwise, we say that $M$ is
protected.
\end{definition}
The notion of protection can be intuitively understood as
following. If we find such $M'$ in $ S_D$ satisfying
$M'^T\sqsubseteq M^T$, then the contraction \be\label{GATC}
    \sum_{B_{i}B_{i+1}\cdots B_j} M^{A_1A_2\cdots A_k}_{B_1B_2\cdots B_k} \left(M'{}_{B_{i}B_{i+1}\cdots B_j}^{D_{i}D_{i+1}\cdots D_j}\right)^*,\quad 1\leq i\leq j\leq k
\ee can be simplified under the constraint
$S_c=\left\{M_s,M_t\right\}$. Diagrammatically, the tensor chain
$M$ becomes disconnected under the contraction with the tensor
chain $M'^*$, as illustrated in Fig.\ref{Figunprotected}. In other
words, when we say a tensor chain is protected, it means that one
can not factorize it by contracting its lower indexes with any
$M\in S_D$ derived from $S_c$. Actually, the condition in
Definition \ref{definitionProtect}, namely ``$\exists M'\in
S_D$'', can be simplified as ``$\exists M'\preceq M_t$''.

\begin{figure}
  \centering
  \includegraphics[width=500pt]{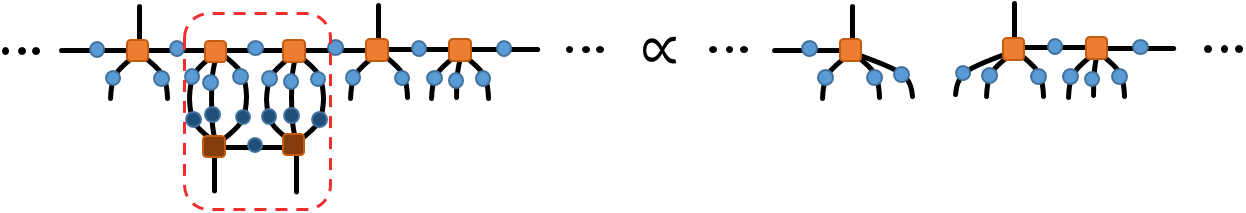}\\
\caption{A unprotected tensor chain $\bpm \cdots&1&0&0&1&0&\cdots
\\ \cdots&2&3&3&2&3&\cdots \epm$ becomes disconnected when
contracting it with the conjugation of $\bbm 1&1\\3&3\ebm$
in $S_c=\left\{\bbm 1\\4\ebm, \bbm
1&1\\3&3\ebm\right\}$.}\label{Figunprotected}
\end{figure}

\subsection{CP tensor chains}\label{SubSectionCP}
In this subsection we point out that given a tiling and $S_c$,
there exists a tensor chain which is critically
protected. We notice that whether a tensor chain $M$ is protected
or not can  be reflected by the value of interior angles, which
roughly speaking measures the curvature of the skeleton of the
tensor chain. Specifically, the larger is $\kappa_{max}(M^T)$,
the easier it is for $M$ to become unprotected. Therefore, there
is a critical value for $\kappa$ at which tensor chain is
critically protected.
\begin{definition}\label{definitionPeriod}
Given an open tensor chain $M=\bbm m_1&m_2&\cdots&m_{k-1}&m_k\\
n_1&n_2&\cdots&n_{k-1}&n_k \ebm$, we can define a periodic tensor chain by joining infinitely many $M$s as
\be
M_{period}=\kc{ \begin{array}{cccccccccccc}
\cdots&m_1-1&m_2&\cdots&m_{k-1}&m_k&m_1-1&m_2&\cdots&m_{k-1}&m_k&\cdots \\
\cdots&n_1&n_2&\cdots&n_{k-1}&n_k-1&n_1&n_2&\cdots&n_{k-1}&n_k-1&\cdots
\end{array} }
\ee
with a loop body $\bpm m_1-1&m_2&m_3&\cdots&m_k\\
n_1&n_2&n_3&\cdots&n_k-1\epm$
\footnote{Here we have exceptionally
    used the notation of closed tensor chain to denote a periodic tensor chain and a loop body, because both of them satisfy (\ref{Closedsummn}).}.
$k$ is called the period of $M_{period}$.
\end{definition}
Obviously, $\kappa(M_{period})=\kappa(M)$.
\begin{definition}\label{definitionCPTC}
Given a tiling and the set $S_c$, we define the critically
protected (CP) tensor chain $M_c$ as the periodic tensor chain
generated by $M_t$. We further define the CP reduced interior
angle as $\kappa_c=\kappa(M_c)=\kappa(M_t)$.
\end{definition}

We demonstrate the construction of $M_c$ with an example in
Fig.\ref{Figloopbody}.

The exact meaning of critical protection is characterized by the following theorem.
\begin{theorem}\label{theoremkappac}
With a given $\{b,a\}$ tiling and a given $S_c$, an open tensor chain $M=\bbm *&*&\cdots&*\\n_1&n_2&\cdots&n_k \ebm$ is unprotected if and only if $\exists p,q$ satisfying $1\leq p\leq q\leq k$ such that
\be\label{hless}
    \sum_{i=p}^{q} (a-1-n_i)\leq(q-p+1)\kappa_c-1.
\ee
A closed tensor chain $M=\bpm *&*&\cdots&*\\n_1&n_2&\cdots&n_k \epm$
is unprotected if and only if $\exists p,h$ satisfying $h\leq k$ such that
\be
\sum_{i=1}^{h} (a-1-n_{(p+i-1)\mod k})\leq h\kappa_c-1.
\ee
\end{theorem}
\begin{proof}
We will present the proof for the case of open tensor chain in
detail and claim that it can be applied to closed tensor chain in
parallel. The main difference will be mentioned in the end of
proof.

We first prove the proposition: if $\exists p,q$ satisfying $1\leq p\leq q\leq k$ such that (\ref{hless}) is true, then $M$ is unprotected.
Without lose of generality, we start with the assumption that $\not\exists
p',q'$ with $p< p'\leq q'\leq q $ or $p\leq p'\leq q'< q $ satisfying
$\sum_{i=p'}^{q'} (a-1-n_i)\leq(q'-p'+1)\kappa_c-1$. Otherwise, we
just replace $p,q$ by $p',q'$.

Let $h=q-p+1$ and $x_i=a-1-n_i$, $\forall i$.
There are two cases:

1. If $h=1$, then $a-n_p\leq \kappa_c\Rightarrow a-n_p\leq [\kappa_c]$. Because of
(\ref{constructm}), $\bbm
a-n_p\\n_p \ebm \preceq \bbm [\kappa_c] \\ a-[\kappa_c]\ebm \preceq
M_t$. We know $\bbm
a-n_p\\n_p \ebm^T \sqsubseteq M^T$, then $M$ is unprotected.

2. If $h\geq 2$, then $\forall l\in\{1,2,\cdots,h-1\}$, we have
\bea
    \sum_{i=p}^{p+l-1} x_i > l \kappa_c - 1,   \label{lgreater1} \\
    \sum_{i=p+l}^{q} x_i > (h-l) \kappa_c - 1 .\label{lgreater2}
\eea
From (\ref{hless})(\ref{lgreater1})(\ref{lgreater2}), we have
\bea
    l\kappa_c-1 < \sum_{i=p}^{p+l-1} x_i < l \kappa_c & \Rightarrow & \sum_{i=p}^{p+l-1} x_i=[l\kappa_c] ,\label{sumsFloorkappac1}\\
    h\kappa_c-2 < \sum_{i=p}^{q} x_i \leq h \kappa_c-1 &\Rightarrow & \sum_{i=p}^{q} x_i = [h \kappa_c]-1 .\label{sumsFloorkappac2}
\eea Let the rational number $\kappa_c=u/v$, where $u,v\in
\mathbb{N^+}$ and $u,v$ are coprime. From
(\ref{sumsFloorkappac1}), we know $l\kappa_c \notin \mathbb{Z}$,
then $l/v\not\in \mathbb{Z},\forall l\in \{ 1,2,...,h-1\}$. Thus
$v\geq h$. From (\ref{sumsFloorkappac1})(\ref{sumsFloorkappac2}),
we have $x_{l+p-1}=[l\kappa_c]-[(l-1)\kappa_c]-\delta_{lh}$. We define $M'=\bbm
m_1'&m_2'&\cdots&m_h'\\n_1'&n_2'&\cdots&n_h' \ebm$ where
$n_i'=n_{i+p-1}$. So $M'^T\sqsubseteq M^T$ and $m_i'=[i\kappa_c]-[(i-1)\kappa_c]-\delta_{i1}-1$. Comparing $M'$ with (\ref{constructm}), we have $M'\preceq M_t$ such that $M$ is unprotected.

Now we prove the converse proposition: If $M$ is unprotected, then
$\exists p,q$ satisfying $1\leq p\leq q\leq k$ such that
(\ref{hless}) is satisfied. Suppose that the top tensor chain
$M_t=\bbm m_1'&m_2'&\cdots&m_v'\\n_1'&n_2'&\cdots&n_v' \ebm$,
and $M$ is unprotected when its tensors from the $p$th to the
$q$th are acted on by a tensor constraint $M'$. Since
$M'^T\sqsubseteq M^T$ and $M'\preceq M_t$, $\exists c,d$
satisfying $1\leq c\leq d\leq v$ and $d-c+1=q-p+1=h$ such that
$x_{p+l-1}\leq m_{c+l-1}'+1 - \delta_{l1} -\delta_{lh}$, $\forall
l\in\{1,2,\cdots,h\}$. By using (\ref{constructm}), we further
have \be\begin{split}
  \sum_{i=p}^q x_i  &\leq h-2+ \sum_{i=c}^d m_i' \\
  &= h-2+\sum_{i=c}^d \kc{ [i\kappa_c]-[(i-1)\kappa_c]-1+\delta_{i1}}\\
  &= [d\kappa_c]-[(c-1)\kappa_c]-2+\delta_{c1}  \\
  &=[d \kappa_c]+[(1-c)\kappa_c]-1  \\
  &\leq [h\kappa_c]-1   \\
  &\leq h\kappa_c-1,
\end{split}\ee
where (\ref{trickfloor}) is used.
So (\ref{hless}) is satisfied.

As far as closed tensor chain is concerned, the only difference is
that it has cyclic symmetry with modular $k$, namely $\bpm
*&*&\cdots&*&*\\n_1&n_2&\cdots&n_{k-1}&n_k \epm = \bpm
*&*&\cdots&*&*\\n_2&n_3&\cdots&n_k& n_1 \epm$, then above $x_i$ is
nothing but reduced interior angles, namely $s_i=a-1-n_i=x_i$. The
proposition can be proved with the same algebra.
\end{proof}

From above theorem, we claim that any tensor chain $M$
satisfying $M^T \preceq M_c^T$ is protected; while any tensor
chain $M$ satisfying $M\preceq M_c$ and $M\neq M_c$ is unprotected.

Similarly, thanks to Theorem \ref{theoremprime}, we have a
one-to-one mapping between CP tensor chain $M_c$ and CP
reduced interior angle $\kappa_c$.

Critical protection characterizes the limit of mapping the
information from one side (with upper indexes) to another side
(with lower indexes) with full fidelity. So the physical
correspondence of CP tensor chain is the maximal boundary of the
region where the interior information can be mapped to the
boundary without loss.

\begin{figure}
  \centering
  \includegraphics[height=40pt]{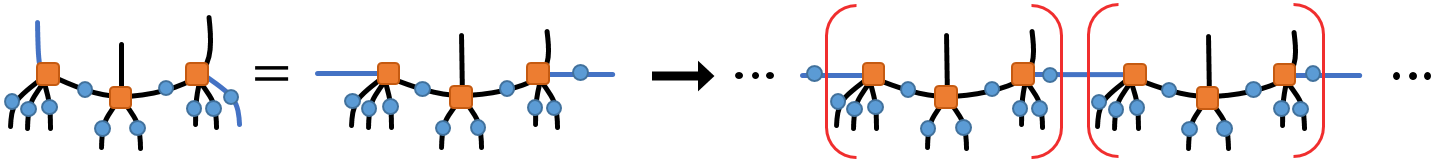}
  \caption{By using top tensor chain $\bbm 1&1&1\\3&2&3 \ebm$, we can define loop body $\bpm 0&1&1\\3&2&2 \epm$ and construct a CP tensor chain $\bpm \cdots&0&1&1&0&1&1&\cdots \\ \cdots&3&2&2&3&2&2&\cdots \epm$.}\label{Figloopbody}
\end{figure}

\section{Geometric description}\label{SectionGeometry}

In this section we elaborate some geometric properties of the
tensor network with $\{b,a\}$ tiling in $H^2$ space, which will
be essential for us to provide a quantitative description for the
QEC and ES in tensor networks. The isometry group of $H^2$ space
is $SL(2,R)$. In $H^2$ space, the curves of constant curvature
(CCC) include circles, hypercircles and horocircles, depending
on their values of geodesic curvature. Geodesic is a kind of
hypercircle \footnote{In many literature, `geodesic' in tensor
    network often refers to the polyline with minimal cuts. However,
    it may not always coincide with the geometrical geodesic in $H^2$
    space. We will not adopt `geodesic' to describe the polyline
    with minimal cuts through this paper.}.
A brief review on $SL(2,R)$ and CCC is given in Appendix \ref{SectionH2}.

\subsection{The curve of constant curvature corresponding to a periodic polyline}\label{SubSectionCCCPL}
The $\ke{b,a}$ tiling
breaks the isometry group $SL(2,R)$ into a discrete subgroup
$G_\text{tiling}$, which is the set of all transformations
preserving the tiling. We are interested in two specific
generators $V,S$ of $G_\text{tiling}$, where $V$ is the
anticlockwise rotation around a node by an angle $2\pi/a$ and $S$
is the clockwise rotation around the midpoint of an edge linked to
this node by $\pi$. $V,S$ should satisfy the following equations
\bea
S^2=V^a=(VS)^b=-1,\\
\Tr(S)=0,\\
\Tr(V)=2\cos(\pi/a),\\
\Tr(VS)=2\cos(\pi/b).
\eea
The solutions up to $SL(2,R)$ are
\be
S=\left(
\begin{array}{cc}
    0 & -1 \\
    1 & 0 \\
\end{array}
\right),\quad
V=\left(
\begin{array}{cc}
    \cos \left(\frac{\pi }{a}\right) & e^{\frac P2} \sin \left(\frac{\pi }{a}\right) \\
    -e^{-\frac P2} \sin \left(\frac{\pi }{a}\right) & \cos \left(\frac{\pi }{a}\right) \\
\end{array}
\right).
\ee
where length $P$ is given in (\ref{lengthp}).

Recall that a tensor chain can be embedded into a tensor network
in $H^2$ space, as discussed at the beginning of Subsection
\ref{SubSectionTensorConstraint}. Similarly, a periodic tensor
chain can also be embedded, whose skeleton forms an endless
and periodic polyline.
When the scale of a chain is much greater than the period of a
polyline, the roughness of the skeleton can be zoomed out such
that it looks like a CCC in $H^2$ space, whose geodesic curvature
$\lambda$ is a constant. In general, given the embedding of a
periodic tensor chain, we can define a unique CCC corresponding to
this chain by specific operation. Next we will firstly present the
procedures to locate such a CCC and finally discuss some
exceptional cases that such CCC could not be defined.

The periodic tensor chain $M$ can be constructed from an open
tensor chain according to Definition \ref{definitionPeriod}. We
choose a node of $M$ and number it by $i$. The loop body beginning
at the $(i+1)$th node is $\bpm
m_{i+1}&m_{i+2}&\cdots&m_{i+k}\\n_{i+1}&n_{i+2}&\cdots&n_{i+k}
\epm$, where $k$ is the period of $M$. Now we choose the center of
the rotation generated by $V$ as the $i$th node and the center of
the rotation generated by $S$ as the midpoint of the edge between
the $i$th node and the $(i+1)$th node. Then we further define a
transformation preserving the structure of periodic polyline as
\be\label{GenM} W= V^{s_{i+k}}S V^{s_{i+k-1}}S \cdots V^{s_{i+2}}S
V^{s_{i+1}}S, \quad s_j=m_j+1, \quad \forall j, \ee which maps
each period in the polyline to the next period along the direction
of the polyline. Starting from a point $q$ in $H^2$ space, the set
of all the points generated by $W^n$, namely $Q_W=\{W^nq|n\in
\mathbb{Z}\}$, will be located on a CCC. When $|Q_W|\geq3$, the
CCC can be uniquely determined.

We are interested in the case that point $q$ is the midpoint of
one edge with lower index in $M$. With some algebra we finally
derive that the geodesic curvature $\lambda$ of this kind of CCC
can be calculated by \be\label{Lambda} \lambda^2=\frac{[\Tr (\pi_q
W)]^2}{(\Tr W)^2+[\Tr(\pi_qW)]^2-4}, \ee where $\pi_q$ is the
matrix of clockwise rotation by angle $\pi$ around point $q$,
which belongs to $SL(2,R)$. $\pi_q$ can be generated by the
generators $V$ and $S$, according to the relative position between
point $q$ and the $i$th node.

Obviously, different choices of point $q$ generate different CCCs
and $\lambda$s. To determine the unique CCC corresponding to $M$,
we remark that one just need to choose the point $q$ which
minimizes $|\lambda-1|$ in (\ref{Lambda}). The process of
generating CCC from a periodic tensor chain is illustrated in
Fig.\ref{FigTCCCC}.

\begin{figure}
    \centering
    \includegraphics[width=0.5\linewidth]{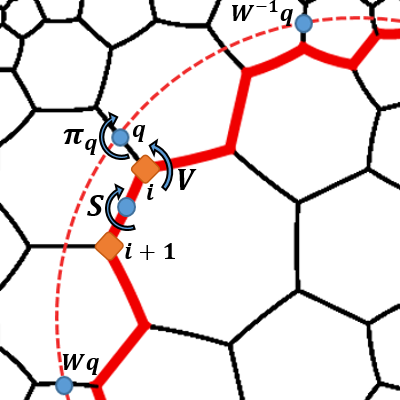}
    \caption{Generate a CCC from a periodic tensor chain with loop body $\bpm 0&1&0 \\ 1&0&1 \epm$ in $\ke{7,3}$ tiling, where $\pi_q=V^{-1}SV$.}
    \label{FigTCCCC}
\end{figure}

Once the CCC corresponding to a periodic tensor chain can be
uniquely determined, the classification of CCC in
(\ref{ClassifyCCC}) can be reformulated by the trace of $W$,
\be
\begin{array}{ll}
    |\Tr(W)|<2,& \text{circle} \\
    |\Tr(W)|=2,& \text{horocircle}\\
    |\Tr(W)|>2,& \text{hypercircle}\\
\end{array}
\ee

From Theorem \ref{theoremprime} and Definition
\ref{definitionPeriod}, for a given rational number
$\kappa\in[1,\frac a2]$, one can construct a unique prime tensor
chain and its periodic tensor chain $M$ whose $\kappa(M)=\kappa$.
As a result, one can further figure out the corresponding CCC as
well as its geodesic curvature $\lambda$. So one can define a
mapping from $\kappa$ to $\lambda$, as illustrated in Fig.\ref{Figkappavslambda}.

It may be noticed that not all the $\lambda$ can be inversely
mapped to $\kappa$, because $\lambda$ is a positive number while
$\kappa$ is a rational number. Nevertheless, horocircle, whose
curvature $\lambda=1$, is a special kind of CCC in $H^2$ space. It
corresponds to a closed tensor chain (closed polyline) in large
radius limit. Furthermore, given a $\{b,a\}$ tiling, we can prove
that the average reduced interior angle $\kappa_h$ of such a
closed tensor chain is \be\label{kappah} \kappa_h=\frac a2 -
\frac{a-2}2\sqrt{\frac{ab-2a-2b}{ab-2a-2b+4}}. \ee Since this
quantity plays a crucial role in classifying the tensor networks,
we provide the detailed proof as follows.
\begin{proof}
    Horocircle is the limit of a circle with
    infinitely long radius in $H^2$ space. To construct the polyline corresponding to
    horocircle in a tensor network, we may consider the process of
    increasing the scale of a closed polyline. In a network with
    $\{b,a\}$ tiling, we consider a closed polyline $M$ with $k$ nodes
    and its total reduced interior angle is $l$. So the average
    reduced interior angle of $M$ is $\kappa=l/k$. Now it is
    helpful for us to plot the Poincare dual of the network with
    $\{b,a\}$ tiling, which is a network with $\{a,b\}$ tiling and all
    the nodes located at centers of dual polygons, as illustrated in
    Fig.\ref{Poincaredual}.

    Given a polyline $M$, its outer adjoint polyline $M'$ in Poincare
    dual network can be constructed step by step: (1) find out
    those elementary polygons in Poincare dual network whose center
    is located at the node of $M$; (2) pick out the edges of
    those polygons outside of $M$; (3) link those edges in order
    and then obtain a closed polyline which is just $M'$.

    The relation between $M$ and $M'$ is illustrated in Fig.\ref{Poincaredual}. One can find that $M'$ has $(a-1)k-l$ nodes and its total reduced interior
    angle is $ak-l$. Keep going on, one can find the outer adjoint
    polyline $M''$ of polyline $M'$ will fall back into the
    network with $\{b,a\}$ tiling. Similarly, polyline $M''$ has
    $(b-1)((a-1)k-l)-ak+l$ nodes and its total reduced interior angle
    is $b((a-1)k-l)-ak+l$. So its average reduced interior angle is
    \be
    \kappa''=\frac{b((a-1)k-l)-ak+l}{(b-1)((a-1)k-l)-ak+l}
    =\frac{(1-b)\kappa+(b-2)c_{1}(a-c_{1})}{(2-b)\kappa+ab-2a-b+1}
    \equiv f(\kappa),
    \ee where \be c_{1}\equiv\frac a2 -
    \frac{a-2}2\sqrt{\frac{ab-2a-2b}{ab-2a-2b+4}}. \ee The number of
    nodes enclosed by $M''$ is always more than that enclosed by
    $M$. Thus, if
    we begin with an elementary closed polyline with $b$ nodes and
    total reduced interior angle $b$ and continuously find the outer
    adjoint polyline, we will approach the polyline corresponding to a
    horocircle. Thus we have \be
    \kappa_h=\lim_{n\to\infty}f^n(1),
    \ee
    where $f^n$ represents $f$ applied $n$ times. To evaluate
    above limit, we observe that
    \be
    \frac{\kappa''-c_{1}}{\kappa''+c_{1}-a}=c_{2} \frac{\kappa-c_{1}}{\kappa+c_{1}-a},
    \ee
    where
    \be
   c_{2}\equiv\frac{1-b-2c_{1}+bc_{1}}{1-b+(2-b)(c_{1}-a)},\quad 0<c_{2}<1.
    \ee
    Thus
    \be
    \lim_{n\to\infty}\frac{f^n(1)-c_{1}}{f^n(1)+c_{1}-a}=0.
    \ee
    Since $f^n(1)$ must be finite, we finally have
    \be
    \kappa_h=\lim_{n\to\infty}f^n(1)=c_{1}.
    \ee
\end{proof}

\begin{figure}
        \includegraphics[width=300pt]{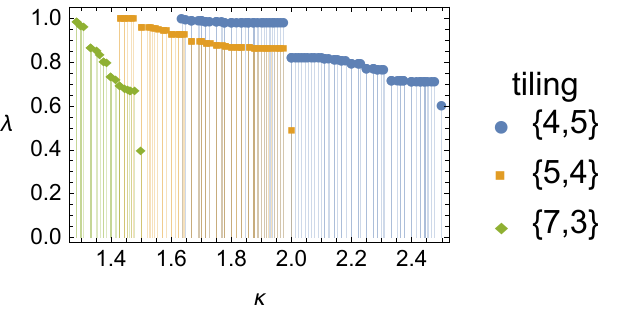}
        \caption{$\lambda$ as a function of $\kappa$ in different tilings, where only some rational numbers of $\kappa\in(\kappa_h,\kappa_0]$ are shown.}
        \label{Figkappavslambda}
\end{figure}

Now we discuss some exceptional cases. The first case is
$|Q_W|\leq 2$, namely the number of generated points less than
two, so CCC can not be uniquely defined through the above process.
 The second case is that $M_{period}$ bends in an irregular way such
that the embedding of periodic tensor chain may lead to a
self-crossed polyline and no CCC could form, such as the
$M_{period}$ with loop body $\bpm 1&1&1\\ 5&4&5 \epm$ in the
tensor network with $\ke{3,7}$ tiling. Such two cases only happen for some $\kappa\in [1,\kappa_h)$.

\subsection{CP curves}
The CCC corresponding to
a CP tensor chain is called CP curve, whose geodesic curvature
is called CP curvature $\lambda_c$. CP curve is a generalization of
the greedy geodesic in \cite{Pastawski:2015qua}.

Given a tiling, CP curvature $\lambda_c$ and CP reduced interior
angle $\kappa_c$ are inversely related to each other, as shown in
Fig.\ref{Figkappavslambda}. If the CP tensor chain form a closed
polyline, the CP curve is a circle, $\kappa_c<\kappa_h$ and
$\lambda_c>1$. If the CP tensor chain form an open polyline
which extends to the boundary, the CP curve is a hypercircle,
$\kappa_c>\kappa_h$ and $\lambda_c<1$.

Roughly speaking, a periodic tensor chain is unprotected if its
corresponding CCC has geodesic curvature $\lambda>\lambda_c$;
while it is protected if the corresponding CCC has
$\lambda<\lambda_c$ .

\begin{figure}
    \centering
    \includegraphics[width=200pt]{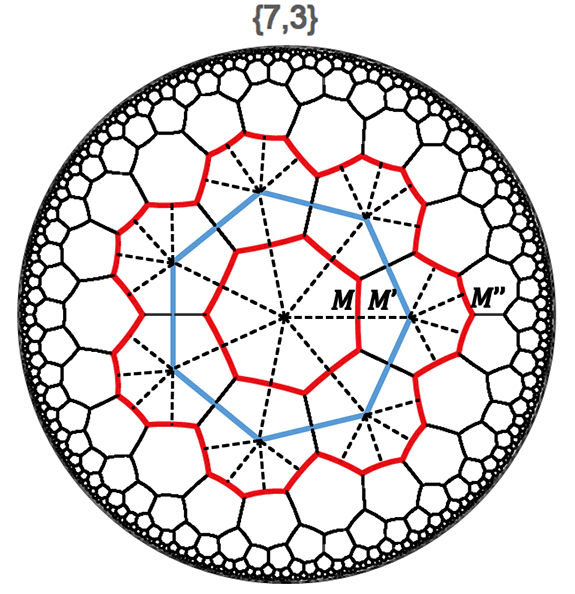}
    \caption{The minimal polyline $M$, the outer adjoint polyline $M'$ of polyline $M$ and the outer adjoint polyline $M''$ of  polyline $M'$ in the network with $\ke{7,3}$ tiling (solid line) and its Poincare dual (dashed line).}\label{FigOuterAdjointPolyline}
    \label{Poincaredual}\end{figure}

\section{Tensor networks and greedy algorithm}\label{SectionTNExamples}

So far we have established a framework to describe tensor
chains and tensor constraints. We will impose the central set $S_c$ to a tensor network in the sense that those constraints $M\in S_D$ derived from $S_c$ are valid, while those $M \notin S_D$ are not valid.

\subsection{Tensor networks with $\ke{7,3}$ tiling or $\ke{4,5}$ tiling}

Before analysing the role of critical protection in QEC and ES for general tensor
networks, we construct some specific examples of tensor networks
and provide an intuitive understanding on the notion of critical
protection.

The structure of a tensor network with $\ke{7,3}$ tiling is shown
in Fig.\ref{Fig73Tiling}. According to (\ref{kappah}), one has
$\kappa_h=1.28$. We impose the central set $S_c=\ke{\bbm 1\\ 2 \ebm, M_t}$ for some typical $M_t$ and discuss
the entanglement property of the tensor network. The structure
of $M_t$, $M_c$, and the values of $\kappa_c$ and $\lambda_c$ are
listed in Table.\ref{Table73}. The corresponding diagrams of
tensor constraints and CP tensor chains in the tiling are
illustrated in Fig.\ref{Fig7311}, \ref{Fig73101}, \ref{Fig731} and
\ref{Fig7310101} respectively.

\begin{table}
    \centering
\begin{tabular}{|c|c|c|c|c|c|c|}
    \hline
    $M_t$ & loop body of $M_c$ & $\kappa_c$ & $\lambda_c$ & CP curve & QEC & ES \\
    \hline
    $\bbm 1\\ 2 \ebm$ & $\bpm 0\\1 \epm$ & $1$ & 1.392 & circle & N & non-flat \\
    \hline
    $\bbm 1&0&1\\1&1&1 \\ \ebm$ & $\bpm 0&0&1\\1&1&0 \epm$ & $4/3$ & 0.869 & hypercircle & Y & non-flat \\
    \hline
    $\bbm 1&0&1&0&1\\1&1&0&1&1 \\ \ebm$ & $\bpm 0&0&1&0&1\\1&1&0&1&0 \epm$ & $7/5$ & 0.738 & hypercircle & Y & mixed \\
    \hline
    $\bbm 1&1\\1&1 \ebm$ & $\bpm 0&1\\1&0 \epm$ & $3/2$ & 0.398 & hypercircle & Y & flat \\
    \hline
\end{tabular}
\caption{Properties about CP and entanglement for different tensor
constraints in $\ke{7,3}$ tiling.}\label{Table73}
\end{table}

\begin{figure}
    \newcommand{\minipagewidth}{0.47\linewidth}
    \newcommand{\figwidth}{110pt}
    \newcommand{\figheight}{30pt}
\begin{minipage}[t]{\minipagewidth}
        \centering
    \includegraphics[height=\figwidth]{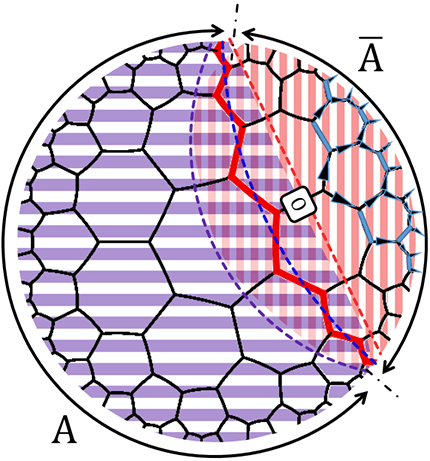}\\    \includegraphics[height=\figheight]{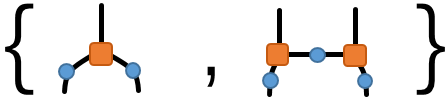}
\caption{The tensor network with $\ke{7,3}$ tiling and
$S_c=\ke{\bbm 1\\2 \ebm, \bbm 1&1\\1&1 \ebm }$. The boundary is
divided into two intervals $A$ and $\bar A$. Those tensors within
the shaded region with purple (red) strips are absorbed by the
greedy algorithm starting from $A$ ($\bar A$). The solid polyline
in red denote a CP tensor chain. The red dashed curve is the CP
curve (hypercircle); the blue dashed curve is its corresponding
geodesic; the purple dashed curve is the reflection of CP curve
with respect to the geodesic. An operator $O$ in the interior is
pushed to the boundary.
    }\label{Fig7311}
\end{minipage}
\hfill
\begin{minipage}[t]{\minipagewidth}
    \centering
  \includegraphics[height=\figwidth]{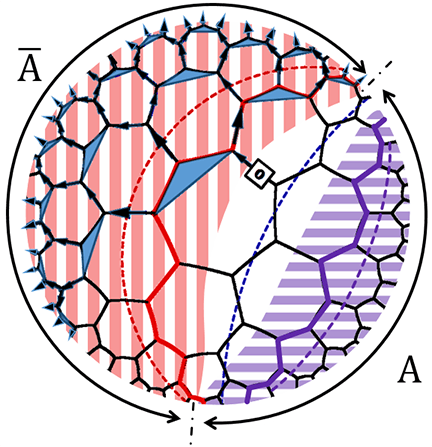}\\
  \includegraphics[height=\figheight]{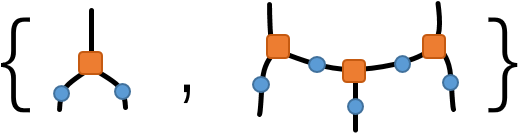}
  \caption{The tensor network with $\ke{7,3}$ tiling and $S_c=\ke{\bbm 1\\2 \ebm, \bbm 1&0&1\\1&1&1 \\ \ebm }$.  An operator is pushed to the boundary. The CP region is
enclosed by two CP curves.  }\label{Fig73101}
\end{minipage}
\vfill
\begin{minipage}[t]{\minipagewidth}
    \centering
    \includegraphics[height=\figwidth]{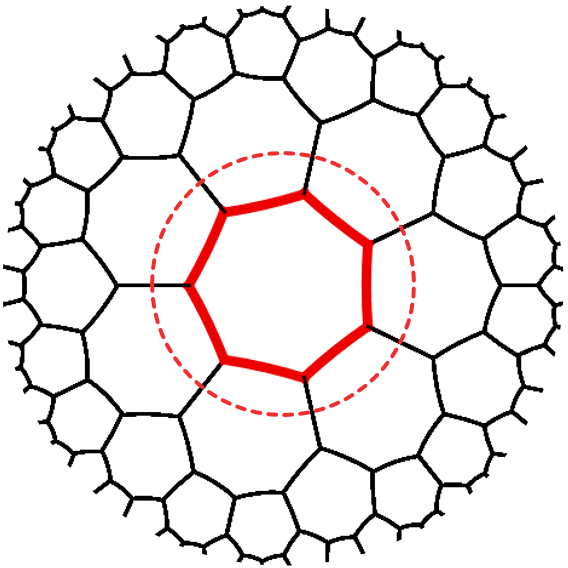}\\
    \includegraphics[height=\figheight]{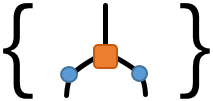}\\
    \caption{The tensor network with $\ke{7,3}$ tiling and $S_c=\ke{\bbm 1\\2 \ebm }$.
    }\label{Fig731}
\end{minipage}
\hfill
\begin{minipage}[t]{\minipagewidth}
    \centering
    \includegraphics[height=\figwidth]{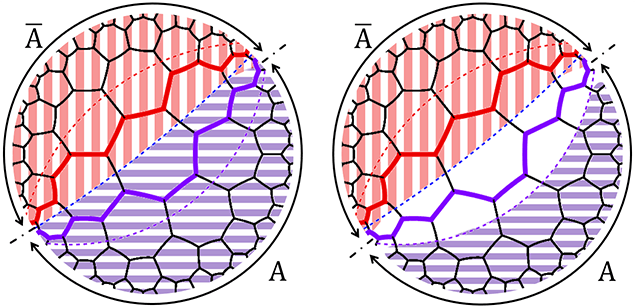}\\
    \includegraphics[height=\figheight]{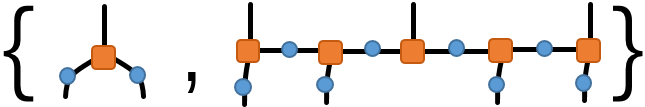}\\
    \caption{The tensor network with $\ke{7,3}$ tiling and $S_c=\ke{\bbm 1\\2 \ebm, \bbm 1&0&1&0&1\\1&1&0&1&1 \\ \ebm }$.
    }\label{Fig7310101}
\end{minipage}
\end{figure}

In parallel, a tensor network with $\ke{4,5}$ tiling is shown in
Fig.\ref{Fig45Tiling}. In this case, one has $\kappa_h=1.63$. The
entanglement properties of this tensor
network with different top tensor chains
are collected in Table.\ref{Table45}. The corresponding diagrams
of tensor constraints and the embedded CP tensor chains are
plotted in Fig.\ref{Fig4522}, \ref{Fig45111}, \ref{Fig451} and
\ref{Fig452} respectively.

\begin{table}
    \centering
    \begin{tabular}{|c|c|c|c|c|c|c|}
        \hline
        $M_t$ & loop body of $M_c$ & $\kappa_c$ & $\lambda_c$ & CP curve & QEC & ES \\
        \hline
        $\bbm 1\\ 4 \ebm$ & $\bpm 0\\3 \epm$ & $1$ & 1.188 & circle & N & non-flat \\
        \hline
        $\bbm 1&1&1\\3&2&3 \\ \ebm$ & $\bpm 0&1&1\\3&2&2 \epm$ & $5/3$ & 0.994 & hypercircle & Y & non-flat \\
        \hline
        $\bbm 2\\3 \\ \ebm$ & $\bpm 1\\2 \epm$ & $2$ & 0.824 & hypercircle & Y & mixed \\
        \hline
        $\bbm 2&2\\2&2 \ebm$ & $\bpm 1&2\\2&1 \epm$ & $5/2$ & 0.604 & hypercircle & Y & flat \\
        \hline
    \end{tabular}
    \caption{Properties about CP and entanglement for different tensor constraints in $\ke{4,5}$ tiling.}\label{Table45}
\end{table}

\begin{figure}
    \newcommand{\minipagewidth}{0.47\linewidth}
    \newcommand{\figwidth}{110pt}
    \newcommand{\figheight}{30pt}
    \begin{minipage}[t]{\minipagewidth}
        \centering
        \includegraphics[height=\figwidth]{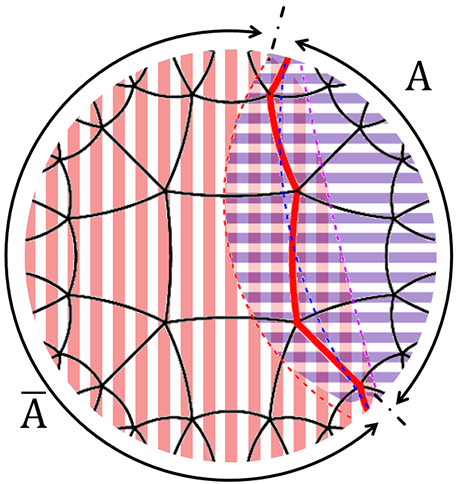}\\
        \includegraphics[height=\figheight]{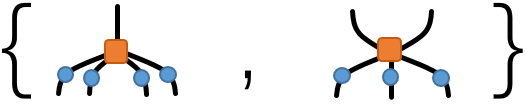}\\
        \caption{The tensor network with $\{4,5\}$ tiling and $S_c=\ke{\bbm 1\\4 \ebm, \bbm 2&2\\2&2 \ebm }$.
        }\label{Fig4522}
    \end{minipage}
    \hfill
    \begin{minipage}[t]{\minipagewidth}
        \centering
        \includegraphics[height=\figwidth]{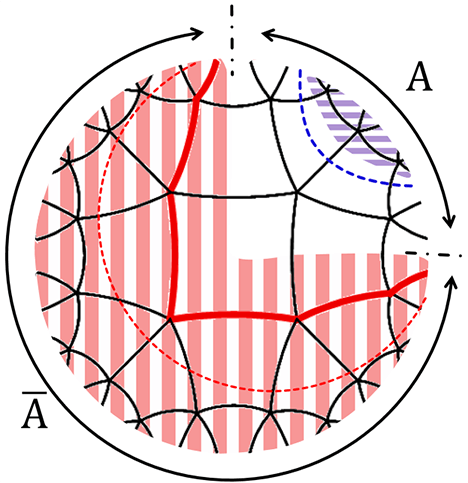}\\
        \includegraphics[height=\figheight]{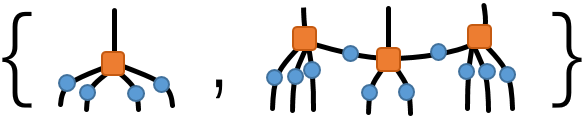}\\
        \caption{The tensor network with $\{4,5\}$ tiling and $S_c=\ke{\bbm 1\\4 \ebm, \bbm 1&1&1\\3&2&3 \ebm }$.
        }\label{Fig45111}
    \end{minipage}
\vfill
    \begin{minipage}[t]{\minipagewidth}
    \centering
    \includegraphics[height=\figwidth]{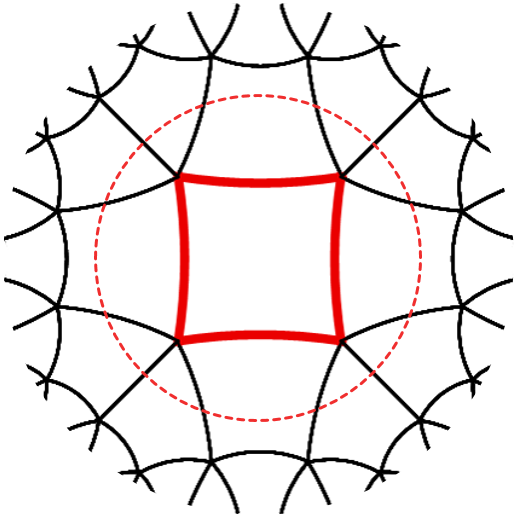}\\
    \includegraphics[height=\figheight]{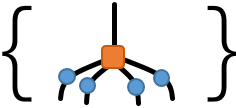}\\
    \caption{The tensor network with $\{4,5\}$ tiling and $S_c=\ke{\bbm 1\\4 \ebm}$.
    }\label{Fig451}
\end{minipage}
\hfill
\begin{minipage}[t]{\minipagewidth}
    \centering
    \includegraphics[height=\figwidth]{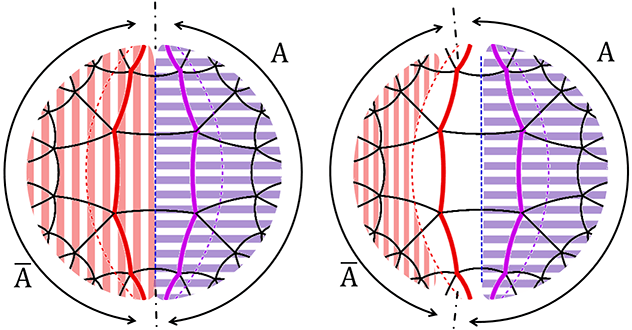}\\
    \includegraphics[height=\figheight]{452Sc.png}\\
    \caption{The tensor network with $\{4,5\}$ tiling and $S_c=\ke{\bbm 1\\4 \ebm, \bbm 2\\3 \ebm }$.
    }\label{Fig452}
\end{minipage}
\label{Fig45}
\end{figure}

In above figures, we divide the boundary of the tensor network
into two intervals $A$ and $\bar A$. The shaded region with
different colors presents the effect of the greedy algorithm
starting from $A$ and from $\bar A$ respectively, which will be
further discussed in next subsection in detail. CP tensor chains
are marked in each figure and its significance in greedy algorithm
will be stressed as well. In next two sections, we will further
take these figures as examples to disclose the relation
between greedy algorithm and quantum error correction as well
as entanglement spectrum.

\subsection{Greedy algorithm on tensor chains}\label{subsectionGA}

For a tensor network $\Psi$, we generalize the greedy algorithm in
\cite{Pastawski:2015qua}, based on the set $S_D$ derived from a
central set $S_c$. After choosing an interval $A$ on the boundary,
we consider a sequence of cuts $\ke{C_n}$ and a sequence of sub
tensor network $\ke{\Phi_n}$, where each $C_n$ is bounded by
$\partial A$ and each $\Phi_n$ consists of those tensors enclosed
by $C_n$ and $A$, shaded with strips. So each $\Phi_n$ is a
mapping from the Hilbert space on $C_n$ to the Hilbert space on
$A$. Let $C_1=A$, then $\Phi_1$ is an identity. Next one figures
out a tensor chain in the tiling which belongs to the set of $S_D$
and all of its lower indexes can be contracted with $\Phi_n$.
Then $\Phi_{n+1}$ is constructed by absorbing such $M_n$
into $\Phi_n$. The greedy algorithm stops when no such tensor
chain can be found. The way of iteration guarantees that each
$\Phi_n$ is proportional to an isometry. As explained in
\cite{Ling:2018qec}, above greedy algorithm for a tensor network
$\Psi$ is equivalent to the procedure of simplifying the
contraction of tensor chains in (\ref{GATC}), where $M$ is any
tensor chain embedded in the tensor network $\Psi$.

To describe the process of greedy algorithm precisely{, which is}
essential in the proofs for the properties of ES, we intend to
extend the notion of protection to a directed cut in greedy
algorithm.

One may notice that process of a greedy algorithm is not unique.
Actually, one may have a lot of ways to arrange the sequence of
absorbing tensor chains into the shaded region $\Phi_n$ such that
during the course of greedy algorithm, $C_n$ need not to be
connected. In the greedy algorithm starting from an interval $A$,
each cut $C_n$ is specified a direction such that its
corresponding $\Phi_n$ is on its right hand side. A cut may
consist of one or more connected components, as illustrated in
Fig.\ref{FigCut}. A connected component can be an open curve or a
closed curve. The open curve is bounded by $\partial A$, while the
closed curve need not. We define the sequence of nodes
corresponding to a connected component as follows.
\begin{definition}
Given a connected component of a directed cut, we find out those
nodes on its left hand side and define the sequence of them along
the direction of the connected component.  The sequence of nodes
given by an open curve is denoted as $[N_1,N_2,\cdots,N_l]$, where
$N_1$ and $N_l$ are located at the boundary. The sequence of nodes
given by a closed curve is denoted as $(N_1,N_2,\cdots,N_l)$,
where $N_1$ and $N_l$ are neighbor.
\end{definition}
Just for convenience, one may allow the sequence number of nodes
to start from any integer, such as $[N_{-1},N_0,N_1,N_2]$. But the
sequence number must be monotonically increasing with unit
step.
\begin{definition}
 Say a tensor chain $M$ is connected to a directed cut $C$, if $M$
lies on the left hand side of $C$ and all the edges associated with its lower indexes are cut by $C$ while all the edges
associated with its upper indexes are not cut by $C$.
\end{definition}
\begin{definition}
Say a directed cut $C$ is unprotected, if there exists
unprotected tensor chain which is connected to $C$. Otherwise, say
$C$ is protected.
\end{definition}
So a greedy algorithm progresses (stops) when the cut is
unprotected (protected).

\begin{figure}
    \centering
    \includegraphics[height=0.5\textheight]{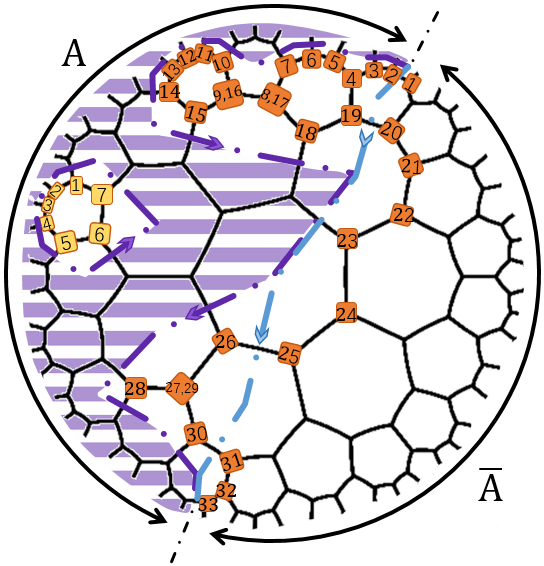}
\caption{A directed cut $C$ during the course of the greedy
algorithm starting from $A$ is marked by purple dot-dashed lines,
which consists of two connected components, an open curve and a
closed curve. The open curve is denoted by
$[N_1,N_2,\cdots,N_{33}]$, with $N_8=N_{17},\,N_9=N_{16}$ and
$N_{27}=N_{29}$. The closed curve is denoted by
$(N_1',N_2',\cdots,N_7')$. The sub tensor networks $\Phi$ absorbed
by the greedy algorithm is shaded with purple strips. A minimal
secant geodesic $G_m$ is marked by blue dot-dashed line.}
    \label{FigCut}
\end{figure}

A greedy algorithm can start from the interval $\bar A$ as well.
In Fig.\ref{Fig7311} \ref{Fig73101} \ref{Fig7310101} \ref{Fig4522}
\ref{Fig45111} \ref{Fig452}, we show the final results of greedy
algorithm on several tensor networks with specific intervals $A$
and $\bar A$, which are shaded with different colors respectively.
In Fig.\ref{Fig731} and \ref{Fig451}, all the tensor chains are
protected under the action of greedy algorithm such that no shaded
region presents in those networks.

In above plots one may notice that some CP tensor chains are
absorbed by greedy algorithm, which apparently conflicts with the
fact that CP tensor chain should be protected. We point out that
this phenomenon ascribes to the fact that the endpoints of CP
tensor chains belong to the interval $A$ or $\bar A$ on the
boundary as well. For instance, consider the greedy algorithm
starting from $\bar A$, as shown in Fig.\ref{FigBoundaryEffect}.
We define $\dot{\bar A}$ to be the interval between one endpoint
of a CP tensor chain, which is an uncontracted edge on the
boundary within $\bar A$, and the most neighboring endpoint of
$\bar A$. Generally, the width of $\dot{\bar A}$ is equal to the
geodesic distance between the CP curve and its axis, {\it i.e.}
$d_c=\text{arctanh}(\lambda_c)$. We firstly consider the greedy
algorithm starting from a sub-interval $\bar A - \dot{\bar A}$. At
this stage greedy algorithm stops before it touches the CP tensor
chain indeed. However, for practice when we calculate the reduced
density matrix $\rho_A$, the contraction on the endpoints of CP
tensor chain, namely uncontracted edges within $\dot{\bar A}$,
must be taken into account by definition. At this stage the CP
tensor chain may fail to be protected under the action of greedy
algorithm, as shown in Fig.\ref{Fig73101}. We refer it as the
boundary effect of greedy algorithm. Nevertheless, we remark that
this boundary effect is weak in the sense that it just absorbs
finite layers (most possibly, only one layer) of tensors, and we
will elaborate it when we study the ES of tensor networks in
Section \ref{SectionES}.

\begin{figure}
    \centering
    \includegraphics[width=0.7\linewidth]{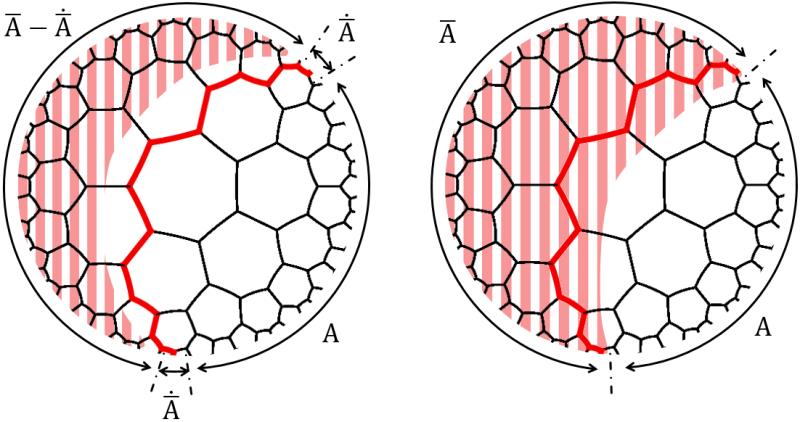}
    \caption{Boundary effect of greedy algorithm in the tensor network with $\ke{7,3}$ tiling and $S_c=\ke{\bbm 1\\2 \ebm, \bbm 1&0&1\\1&1&1 \ebm }$. Interval $\bar A$ is split into $\bar A - \dot{\bar A}$ and $\dot{\bar A}$. The CP tensor chain is denoted by the polyline in red.}
    \label{FigBoundaryEffect}
\end{figure}

\section{Quantum error correction (QEC)}
In this section we will concentrate on how to justify the
ability of QEC for a tensor network based on the properties of CP
tensor chain.

\subsection{Greedy algorithm and QEC}
The whole story of QEC on tensor networks is based on the Hilbert
space associated with uncontracted edges which introduce extra
degrees of freedom in the bulk and the corresponding code subspace
in the Hilbert space on the boundary. The correction to the code
subspace after erasing an interval $A$ on the boundary is
equivalent to pushing a bulk operator in the wedge of the interval
$\bar A$ to the interval $\bar A$ on the boundary
\cite{Almheiri:2014lwa}. Technically, following
\cite{Pastawski:2015qua}, the procedure of QEC involves in three
steps: (1) acting on an uncontracted index in the bulk with an
operator; (2) pushing the operator from this index in the bulk to
the indexes in the network; (3) pushing it to the boundary
further.

In our current work we will ignore the Hilbert space in the bulk
since our main purpose is to realize the algorithm of QEC on
tensor networks. We will skip step 1 and begin at step 2, by
directly inserting an operator into the contracted edges in the
interior. In the language of tensor chain, we can insert an
operator $O$ between tensor chain $M$ and $M'$ \be\label{insert}
\sum_B M^A_B M'^{B}_C \to \sum_{BD} M^A_B O^B_{D} M'^{D}_C. \ee No
matter which way one adopts to insert an operator into the
network, the latter processes of QEC are the same. Thanks to
tensor constraint, we can push an operator $O$ through a tensor
chain $M\in S_D$ and the output operator can be rewritten as $O'$,
namely \be\label{push} \sum_B O^A_BM^B_C \propto \sum_B M^A_B
O'^B_C, \quad O'^B_C= \sum_{AD} (M^B_A)^* O^A_D M^D_C. \ee

Specifically, without uncontracted indexes in the bulk, equation
(\ref{insert}) is just the reflection of step 2 above and equation
(\ref{push}) depicts step 3. After all, the terminology `QEC' in
this paper refers to the above interpretation.

All above operations can be demonstrated by diagrams. Taking
the tensor network with $\ke{7,3}$ tiling as an example. The
insertion of an operator is illustrated as
\be\label{73insert}
\includegraphics[height=40pt]{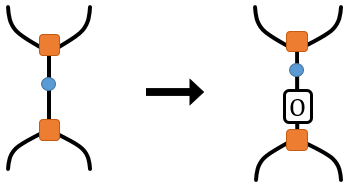}.
\ee
While employing tensor constraints (\ref{Fig7311}), the
process of pushing an operator through tensor chains can be
illustrated as
\bea
\includegraphics[height=40pt]{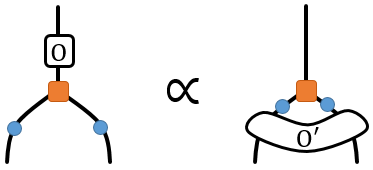},\quad\quad
\includegraphics[height=40pt]{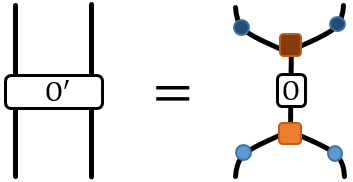}
\label{1to2}    \\
\includegraphics[height=40pt]{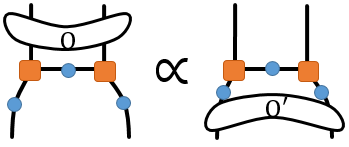},\quad\quad
\includegraphics[height=40pt]{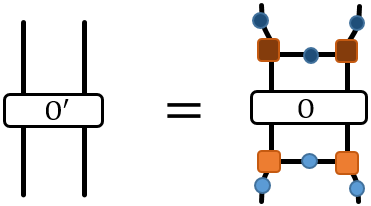}.
\label{11to11}
\eea

One can successively push operators through tensor chains in
$S_D$. Operators may be finally pushed to an interval on the
boundary or not, depending on the structure of tilings and tensor
constraints. Actually, tensor pushing is the reverse procedure of
greedy algorithm, where pushing an operator through tensor chain
$M\in S_D$ is reverse to the procedure of absorbing a tensor chain
$M$ into the shaded region of a tensor network.

\begin{definition}
We say that a tensor network enjoys QEC if any operator inserted
into the bulk of the network can be pushed to an interval on
the boundary.
\end{definition}

In Fig.\ref{Fig7311} and \ref{Fig73101}, an inserted operator $O$
is successfully pushed to $\bar A$. We remark that if the
operator is inserted into the region enclosed by the CP tensor
chain and interval $\bar A$, as illustrated in
Fig.\ref{Fig7311}, then it can be pushed to $\bar A$. In other
word, after erasing an interval $A$, most of those points in the
wedge of $\bar A$ can be recovered by QEC.

While if the operator is inserted into the region enclosed by the
CP tensor chain and the geodesic bounded by $\partial \bar A$, the
situation becomes subtle and it is not guaranteed that the
operator can always be pushed to $\bar A$. On one hand, if the
inserted operator is close to CP tensor chain, as illustrated
in Fig.\ref{Fig73101}, then it may still be pushed to a
subinterval in $\bar A$. While, now the bound of such
subinterval is approaching to $\partial \bar A$. In this
figure we notice that a lot of arrows, which denote the
trajectory of pushing the operator through, go across the geodesic
and then radiate out in a wide region, in contrast to the
process in Fig.\ref{Fig7311}. Such phenomenon indicates that the
information of an operator can only be recovered in a wide range
of the boundary, implying the function of QEC in
Fig.\ref{Fig73101} is weaker than that in the tensor network in
Fig.\ref{Fig7311}. It may be related to the approximate QEC
\cite{Schumacher:2002aqe,Almheiri:2014lwa}. On the other hand, if
the operator is rather close to the geodesic bounded by $\partial
\bar A$, it may not be pushed to $\bar A$ any more.

\subsection{CP curves and QEC}

The geometric description of CP tensor chain in Section
\ref{SectionGeometry} provides us a way to describe QEC over $H^2$
space as well. Given a subsystem $\bar A$ on the boundary, one may
ask whether an operator acting on point $x$ which locates inside
the wedge of $\bar A$ can be pushed to $\bar A$. For a simply
connected interval $\bar A$, we denote its two endpoints as $u$
and $v$, respectively. Then these two points together the point
$x$ can uniquely determine a hypercircle $H$ in $H^2$ space. The
sub network in region $\Omega$ enclosed by $H$ and $\bar A$
defines a mapping $\Phi$ from the Hilbert space associated with
the edges on $H$ to the Hilbert space associated with the edges on
$\bar A$. If and only if $\Phi$ is proportional to an isometry,
then the operator can be pushed to the boundary, thus implementing
QEC in an operator scenario. Otherwise the operator can not be
pushed into the region specified by $\bar A$, and the recovery of
such an operator will be prevented by erasing $A$ such that QEC
fails.

To check whether $\Phi$ is isometric or not, one needs to evaluate
the inner product $\Phi \Phi^\dagger$, which is directly
determined by the imposed tensor constraints $S_c$. During the
evaluation process, the most difficult step is to simplify the
contraction $M M^\dagger$, where $M$ is the boundary tensor chain
of $\Phi$ on $H$. So to figure out whether $\Phi$ is isometric or
not, our final task is to justify whether $M$ is protected or not
under tensor contractions which are subject to $S_c$.

Fortunately, our discussion in the section of critical
protection has provided an answer to this question. One can
justify this by comparing the geodesic curvature $\lambda$ of the
hypercircle $H$ with the the curvature of CP curve $\lambda_c$. If
$\lambda>\lambda_c$, then $M$ is unprotected; if
$\lambda<\lambda_c$, then $M$ is protected.

As a result, given a subsystem $A$ on the boundary, we find a geodesic connecting two end points of the subsystem $A$ and a CP curve between the geodesic and $\bar A$. Whether an
operator at $x$ can be pushed into $\bar A$ depends on the geodesic
curvature of the hypercircle passing through $x$. For those points
inside the region enclosed by boundary $\bar A$ and the CP curve, an operator can be recovered by QEC since the geodesic curvature of hypercircles is greater than $\lambda_c$; while for those points inside the region enclosed by the CP curve and the geodesic an operator can not be recovered by QEC since the
geodesic curvature is less than $\lambda_c$.

When a tiling of $H^2$ is specified, $\kappa_c$ is inversely
related to $\lambda_c$. Because $\kappa$ is more easily calculated
than $\lambda$, one can alternatively compare the average reduced
interior angle $\kappa$ of a hypercircle with the average reduced
interior angle of CP tensor chain $\kappa_c$.  For a given tiling,
recall that the tensor chain corresponding to a horocircle with
$\lambda_h=1$ has average reduced interior angle $\kappa_h$ in
(\ref{kappah}). Moreover, once $S_c$ is specified, then
$\lambda_c$ and $\kappa_c$ are determined as well. Whether a
tensor network enjoys QEC or not can be justified by comparing the
value of $\lambda_c$ with $\lambda_h$ or $\kappa_c$ with
$\kappa_h$, as described below.

If $\lambda_c\geq1$ or $\kappa_c\leq\kappa_h$, the CP curve is a
circle or a horocircle. The geodesic curvature of all hypercircles
must be less than $\lambda_c$, so no QEC can be implemented by
inserting an operator into any point in the bulk and such a tensor
network do not enjoy QEC. For instance, those tensor networks in
Fig. \ref{Fig731} and \ref{Fig451} belong to this class. It
matches the fact that greedy algorithm does not iterate in these
tensor networks.

Similarly, if $\lambda_c<1$ or $\kappa_c>\kappa_h$, the CP curve
is a hypercircle. An operator inserted into the region
enclosed by the CP tensor chain and $\bar A$ can be recovered by
QEC and such a tensor network enjoy QEC. For instance, all the
other tensor networks except Fig.\ref{Fig731} and \ref{Fig451} in
this paper belong to this class. Nevertheless, given an interval
$\bar A$, the region that can be recovered is different for
different constraints.

We summarize the above results about the function of QEC in a
network in Fig.\ref{FigLambdaKappa}.

\section{Entanglement spectrum (ES)}\label{SectionES}

Next we focus on the evaluation of entanglement spectrum for a
given tensor network, and argue that the flatness of ES can be
justified with the power of critical protection in general
cases.

\subsection{Reduced density matrix}

A tensor network $\Psi$ gives a state $\ket{\Psi}$ in the
Hilbert space defined on its uncontracted edges on the
boundary. Given an interval $A$ on the boundary, one can obtain
the reduced density matrix of $A$ by tracing out the
complementary region $\bar A$, namely
\be\label{rhoA0}
    \rho_A=\Tr_{\bar A}\ket{\Psi}\bra{\Psi}.
\ee

We are concerned with the issue whether the reduced density
matrix $\rho_A$ has a flat spectrum, which means that all the
non-zero eigenvalues of $\rho_A$ are identical. This statement
can be alternatively rephrased as the following propositions:
\begin{itemize}
    \item
All the orders of Renyi entropy
\be\label{RenyiEntropy}
S_{A,n}=-\frac1{n-1}\ln\frac{\Tr\kc{\rho_A^n}}{\kc{\Tr\rho_A}^n}
\ee
are identical, namely independent of $n$.
    \item
Reduced density matrix satisfies the relation
\be\label{flatES}
\rho_A^2\propto\rho_A.
\ee
\end{itemize}

As indicated at the beginning of this paper, $\rho_A$ of the
ground state of $CFT_2$ satisfies (\ref{CardyFormula}) and
exhibits a non-flat ES. The gravitational dual result of $AdS_3$
vacuum coincides with the above result as well. Now we would like
to check whether the ES of a tensor network state is flat or not.
For this purpose it is convenient to check the relation in
(\ref{flatES}) by manipulating tensor networks.

First we disclose the key role of CP tensor chain in identifying
the protected region in a tensor network. Recall the boundary
effect in greedy algorithm, we intend to separate the procedure of
taking trace on $\bar A$ into following two steps
 \be\label{rhoA}
\rho_A=\Tr_{\bar A}\ket{\Psi}\bra{\Psi}=\Tr_{\dot{\bar
A}}\Tr_{\bar A-\dot{\bar A}}\ket{\Psi}\bra{\Psi}. \ee

In the following, we will take tensor networks with $\ke{7,3}$
tiling as examples to demonstrate the evaluation of ES by
manipulating tensor networks. The results have previously been
collected in Table \ref{Table73}.

\subsubsection{Non-flat ES}

First of all, we point out the evaluation of ES depends on the
choice of the interval $A$ on the boundary. We will see that, for
constraints $S_c=\ke{\bbm 1\\2 \ebm, \bbm 1&0&1\\1&1&1 \ebm}$,
the ES of $\rho_A$ for any relatively large interval $A$ is
non-flat. We call the tensor network generally has a non-flat ES. The word ``generally'' means that the ES is always non-flat unless fine-tuning tensor $T$ and $E$. Throughout this paper when we say that a tensor network has a non-flat ES, we refer to the above statement.

Firstly, we trace out the degrees of freedom in $\bar A-\dot{\bar
A}$ to obtain the reduced density matrix. During this procedure
the structure of tensor network is simplified due to the tensor
constraints generated by $S_c$, see Fig.\ref{Fig73101GA} (a-c).
Specifically, those tensors in the wedge of $\bar A$ are
contracted into identity matrices, which is just the process of
the greedy algorithm starting from $\bar A - \dot{\bar A}$ in the
previous section. One can repeatedly consider this process until
it reaches a final stage that the network can not be simplified
any more, as shown in Fig.\ref{Fig73101GA}(c). The terminal
boundary forms a polyline in red as marked in
Fig.\ref{Fig73101GA}. As a matter of fact, such a polyline is
nothing but a CP tensor chain as we have defined in previous
section. From this figure we perceive that, before the trace of
uncontracted edges in $\dot{\bar A}$ is taken into account, the
operation induced by the greedy algorithm can not enter the region
enclosed by CP tensor chain and $A$, which is exactly the reason
why we call it critically protected tensor chain.

Now the next step is to evaluate $\Tr_{\dot{\bar A}}$, namely
tracing the degrees of freedom associated with uncontracted edges
on the boundary which are mostly neighboring to $A$. The process
is illustrated in Fig.\ref{Fig73101GA}(d)(e). We notice that the
network structure can be further simplified such that CP tensor
chains are absorbed into the shaded region at this step, which
is the boundary effect of greedy algorithm as we described in
previous section.

The boundary effect of greedy algorithm results from the
discretization of $H^2$ space, which may not appear in a
continuous geometry. In the context of tensor networks, however,
according to (\ref{rhoA}) the uncontracted edges in $\dot{\bar A}$
should be contracted. Sometime the contribution of this effect to
reduced density matrix becomes subtle, and we should cautiously
handle this effect. In other words, whether the ES is flat or not
can only be justified after the boundary effect is taken into
account.

Now with the reduced density matrix $\rho_A$ at hand, we can
compute $\rho_A^2$ by further contracting those uncontracted edges
in $A$. One can simplify $\rho_A^2$ by virtue of tensor
constraints, which is parallel to the above process on $\bar A$.
Boundary effect of greedy algorithm appears as well. Before the
boundary effect is taken into account, the greedy algorithm stops
at a CP tensor chain, which is the reflection of the CP tensor
chain appearing in the contraction on $\bar A$ about the geodesic
bounded by $\partial A$. Thus, the simplification of $\rho_A^2$ is
equivalent to applying the greedy algorithm to $A$ and $\bar A$
successively, as shown in Fig.\ref{Fig73101}.

The calculation of $\rho_A^2$ is demonstrated in
Fig.\ref{Fig73101ES}. Obviously from this diagram we find that
$\rho_A^2$ can not be simplified to be proportional to $\rho_A$
such that equation (\ref{flatES}) is not satisfied. Equivalently,
from Fig.\ref{Fig73101}, we notice that some tensors are not
absorbed by the greedy algorithm starting from $A$ and $\bar A$,
thus (\ref{flatES}) is not satisfied.

Given the above constraints, we point out that as long as $A$
is large enough, $\rho_A$ always gives rise to a non-flat ES,
independent of the choice of $A$. This assertion will be proved in
Subsection \ref{SubsectionCPES}. Right now we just conclude that
such a tensor network has a non-flat ES, in agreement with what is
found by explicitly computing the eigenvalues of the reduced
density matrix in \cite{Evenbly:2017htn}.

We remark that for the above constraints the corresponding CP
curve is a hypercircle. When the CP curve is a horocircle, it
approaches the boundary with single intersecting point. Or when
the CP curve is a circle, it does not reach the boundary. For both
cases one need not consider the boundary effect separately, and
the ES is usually non-flat for CP circles since the region
enclosed by the circle is protected.

\begin{figure}
    \newcommand{\length}{120pt}
    \newcommand{\distance}{0pt}
    \centering
    \subfigure[]{\label{Fig73101GA1}
        \includegraphics[height=\length]{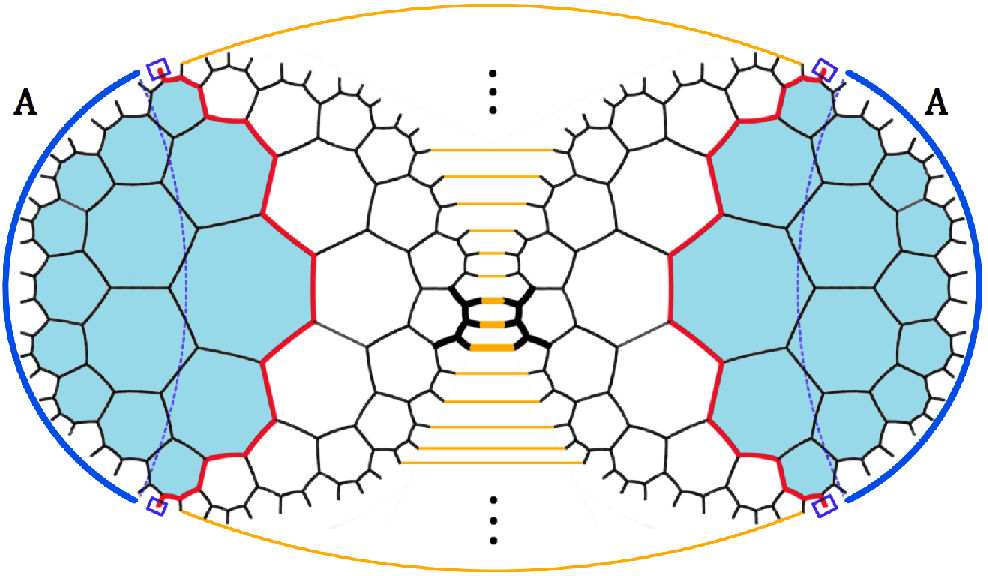}}
    \hspace{\distance}
    \subfigure[]{\label{Fig73101GA2}
        \includegraphics[height=\length]{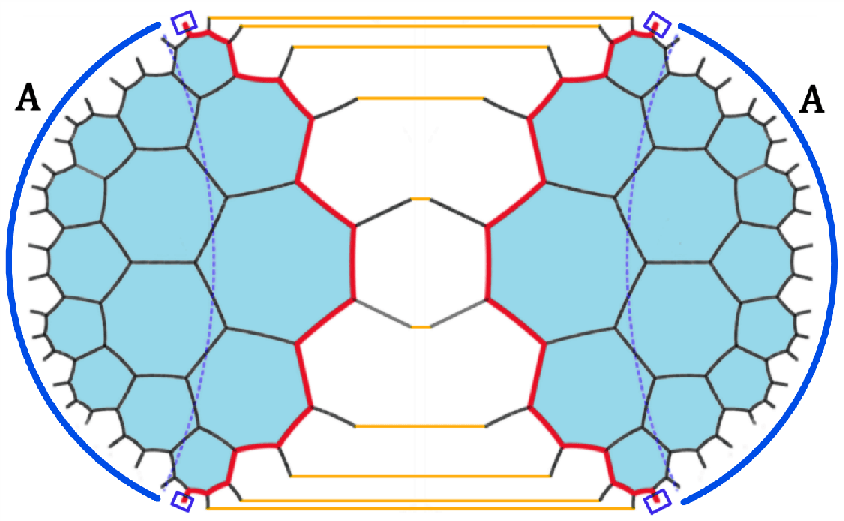}}
    \hspace{\distance}
    \subfigure[]{\label{Fig73101GA3}
        \includegraphics[height=\length]{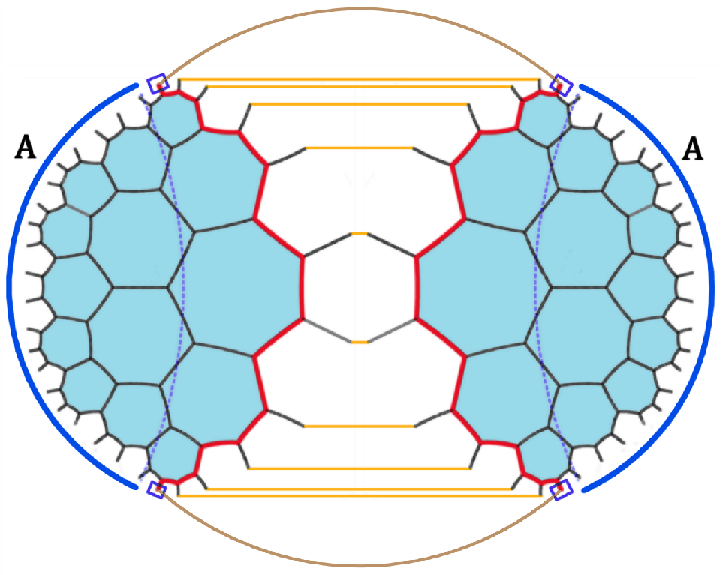}}
    \hspace{\distance}
    \subfigure[]{\label{Fig73101GA4}
        \includegraphics[height=\length]{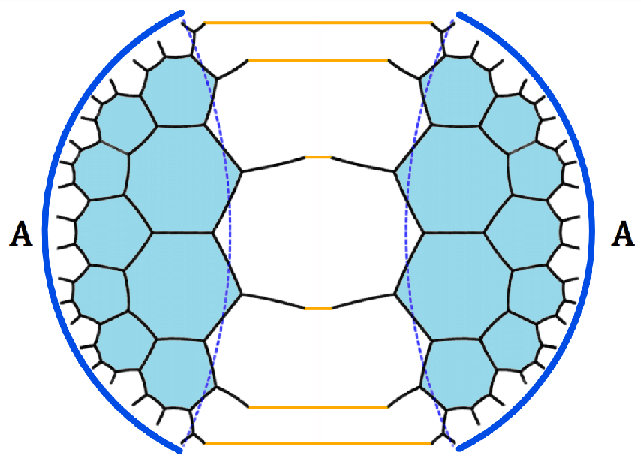}}
    \caption{The evaluation of $\rho_A$. The interval $A$ is marked by solid lines in blue. The geodesic bounded by $\partial A$ is plotted as a dashed line in blue. The CP tensor chain is marked by a solid line in red, and its two endpoints on the boundary are marked by two
        blue rectangles. (a) The contractions of indexes in $\bar A$
        are illustrated by orange lines. (b) Before the endpoints of
        CP tensor chain are contracted, CP tensor chains are not absorbed. (c)
        The contractions of the endpoints of CP tensor chain on the boundary is marked by brown lines. (d) CP tensor chains are absorbed at the final stage.}\label{Fig73101GA}
\end{figure}

\begin{figure}
    \includegraphics[height=100pt]{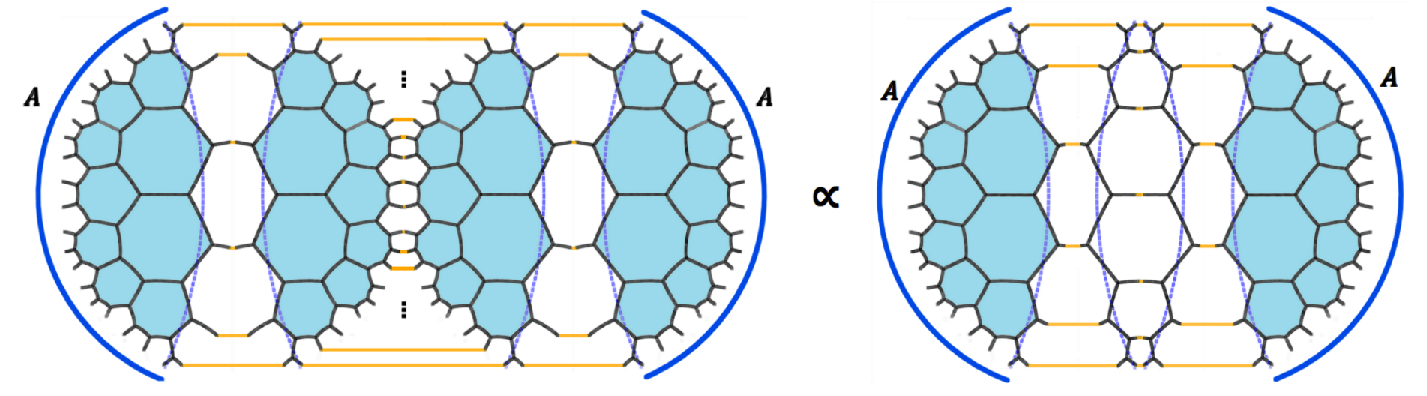}
    \caption{The diagram for the evaluation of $\rho_A^2$. (\ref{flatES}) is not satisfied.}\label{Fig73101ES}
\end{figure}

\subsubsection{Flat ES}

We have pointed out that one equivalent way to check the relation
in (\ref{flatES}) is to consider the greedy algorithm starting
from $A$ and from $\bar A$ successively. Let us take Fig.\ref{Fig7311}
as an example, where $S_c=\ke{\bbm 1\\2 \ebm, \bbm 1&1\\1&1
\ebm}$. We observe that the union of these two shaded regions
covers the whole tensor network, implying that all the tensors are
absorbed by the greedy algorithm. Therefore, (\ref{flatES}) is
satisfied and ES has to be flat. We call the tensor network has a
flat ES.

\subsubsection{Mixed ES}

From Fig.\ref{Fig7310101}, we know that, for $S_c=\ke{\bbm 1\\2
\ebm, \bbm 1&0&1&0&1\\1&1&0&1&1 \\ \ebm}$, the ES of $\rho_A$ can
be flat or non-flat, depending on the choice of $A$. We call
the tensor network has a mixed ES.

\subsection{Geometric point of view on ES}\label{SubSectionESCPCurve}

In the tensor network realization of AdS/CFT, a tensor network is
usually treated as the wavefunction $\Psi$ of the ground state.
Alternatively, when an interval $A$ on the boundary is given, we
notice that $\Psi$ can be understood as a mapping from the Hilbert
space on $A$ to the Hilbert space on $\bar A$. So $\Psi$ can be
regarded as a matrix $\Psi^A_{\bar A}$, where two indexes $A$ and
$\bar A$ represent the degrees of freedom on two subsystems $A$
and $\bar A$, respectively.

The notion of critical protection provides us an efficient way to
visualize the simplification of tensor networks under the tensor
contractions which are subject to tensor constraints. To make this
process more transparent, we firstly intend to decompose a network
into some sub networks. As seen in previous subsections, when the
indexes on $A$ or $\bar A$ are contracted, the greedy algorithm
will stop at some nodes. Let us firstly neglect the boundary
effect, then the skeletons of connecting those nodes will form two
CP tensor chains, which are neighboring to the geodesic bounded
by $\partial A$.

First of all, we point out that when $\lambda_c\geq1$ or
$\kappa_c\leq\kappa_h$, all the hypercircles are protected since
their geodesic curvatures are less than $\lambda_c$. So a non-flat
ES is guaranteed. In the following, we will focus on the
non-trivial case, $\lambda_c<1$ or $\kappa_c>\kappa_h$, where CP
curves are hypercircles.

We denote the CP curve close to $A$ or $\bar A$ as $H_A$ or
$H_{\bar A}$, respectively. The region enclosed by two CP curves
is called CP region $\Omega_c$. Those tensors in the CP region
form a sub tensor network $\Psi_c$, which is a mapping from $H_A$
to $H_{\bar A}$ and is denoted as $(\Psi_c){}^{H_A}_{H_{\bar A}}$.
Similarly, $H_A$ and $A$ enclose a sub tensor network $\Phi_A$,
which defines a mapping $(\Phi_A){}^A_{H_A}$; $H_{\bar A}$ and
$\bar A$ enclose a sub tensor network $\Phi_{\bar A}$, which
defines a mapping $(\Phi_{\bar A}){}^{\bar A}_{H_{\bar A}}$. Since
the tensors outside $H_A$ are not protected under the contraction
of $A$, the mapping $(\Phi_A){}^A_{H_A}$ from $H_A$ to $A$ should
be proportional to an isometry. Similarly, the mapping
$(\Phi_{\bar A}){}^{\bar A}_{H_{\bar A}}$ from $H_{\bar A}$ to
$\bar A$ is proportional to an isometry as well. It is denoted as
\be\label{PsiIso}
    \Phi_A^\dagger\Phi_A\propto I, \quad \Phi_{\bar A}^\dagger\Phi_{\bar A}\propto I',
\ee
where the indexes are abbreviated and $I$ ($I'$) is identity matrix on $A$ ($\bar A$).

Finally, the full matrix $\Psi^A_{\bar A}$ can be represented as the product of matrices
\be
    \Psi=\Phi_A\Psi_c\Phi_{\bar A}^\dagger.
\ee
Then it is easy to see
\bea
    \rho_A&=&\Psi\Psi^\dagger= \Phi_A\Psi_c\Phi_{\bar A}^\dagger\Phi_{\bar A}\Psi_c^\dagger\Phi_A^\dagger \propto \Phi_A\Psi_c\Psi_c^\dagger\Phi_A^\dagger , \\
    \rho_A^2&=&\Psi\Psi^\dagger\Psi\Psi^\dagger\propto \Phi_A\Psi_c\Psi_c^\dagger\Psi_c\Psi_c^\dagger \Phi_A^\dagger ,
\eea
where (\ref{PsiIso}) is used. A flat ES in (\ref{flatES}) means that
\be\label{flatESCP}
    \Psi_c\Psi_c^\dagger\Psi_c\Psi_c^\dagger    \propto     \Psi_c\Psi_c^\dagger .
\ee
We present a schematic diagram to demonstrate the decomposition of tensor network state as well as the calculation of $\rho_A$ and $\rho_A^2$ in Fig.\ref{FigrhoA2}. The condition for flat ES (\ref{flatESCP}) is illustrated in Fig.\ref{FigFES}.
This figure reveals that whether the ES is flat or not depends on the thickness of the CP region where the thickness of the CP region is defined by the distance between the two CP curves.

Equivalently, from above derivation we notice that the flatness of
ES may be checked by observing the result of the greedy algorithm
starting from $A$ and from $\bar A$ successively, which
figures out the region of isometry between tensor chains in
(\ref{PsiIso}). If all the tensors are absorbed by the greedy
algorithm, then (\ref{flatESCP}) is valid and the ES is flat, and
vice versa.

Once the boundary effect is considered, as we showed in previous
section, CP tensor chains on the boundary of the CP region
$\Omega_c$ are not protected any more under the greedy algorithm.
Nevertheless, only a finite thickness of the CP region will be
absorbed. In Fig.\ref{Fig73101ES}, since those tensors close to
the geodesic are not absorbed, the tensor network has a non-flat
ES.

\begin{figure}
\flushleft
  \includegraphics[height=70pt]{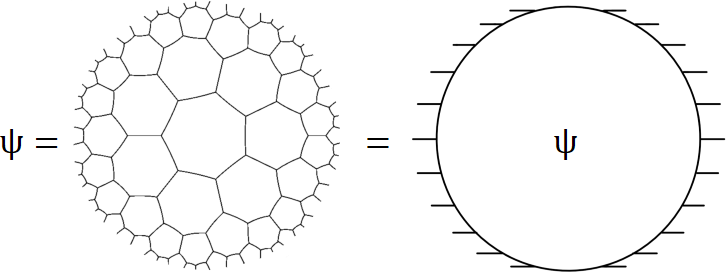}
  \includegraphics[height=70pt]{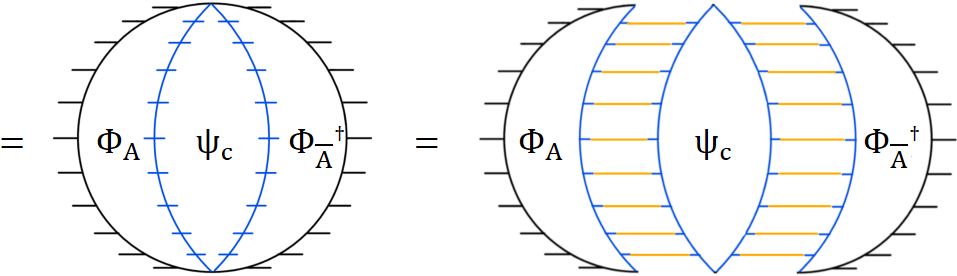}\\
  \includegraphics[height=70pt]{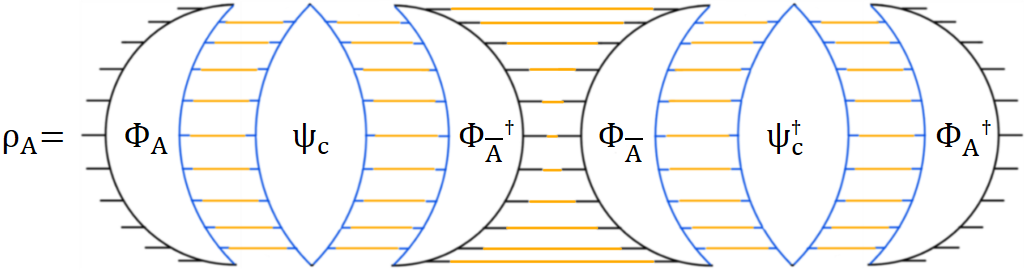}
  \includegraphics[height=70pt]{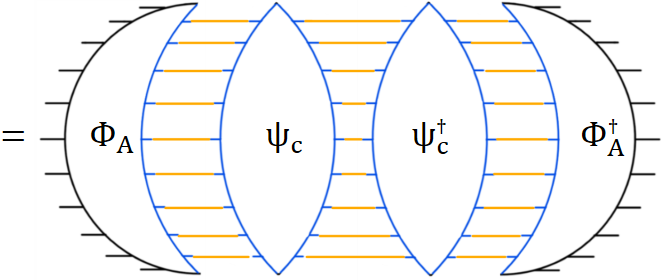}\\
  \includegraphics[height=68pt]{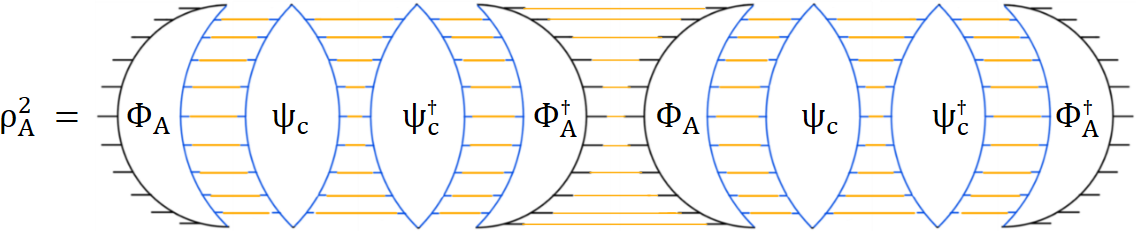}\\ \quad\,
  \includegraphics[height=70pt]{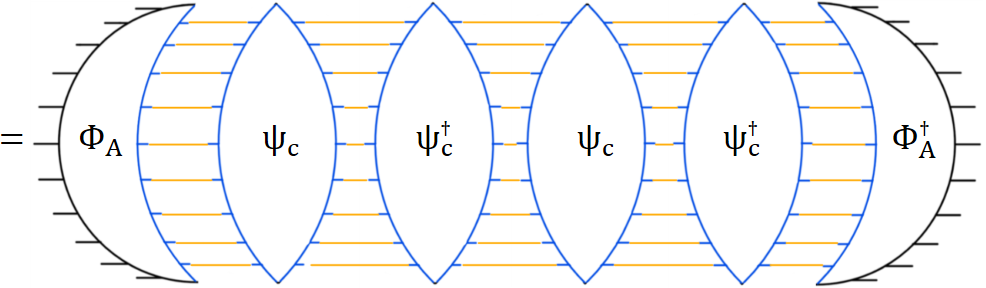}
  \caption{Calculating $\rho_A$ and $\rho_A^2$.}\label{FigrhoA2}
\end{figure}

\begin{figure}
  \centering
  \includegraphics[height=80pt]{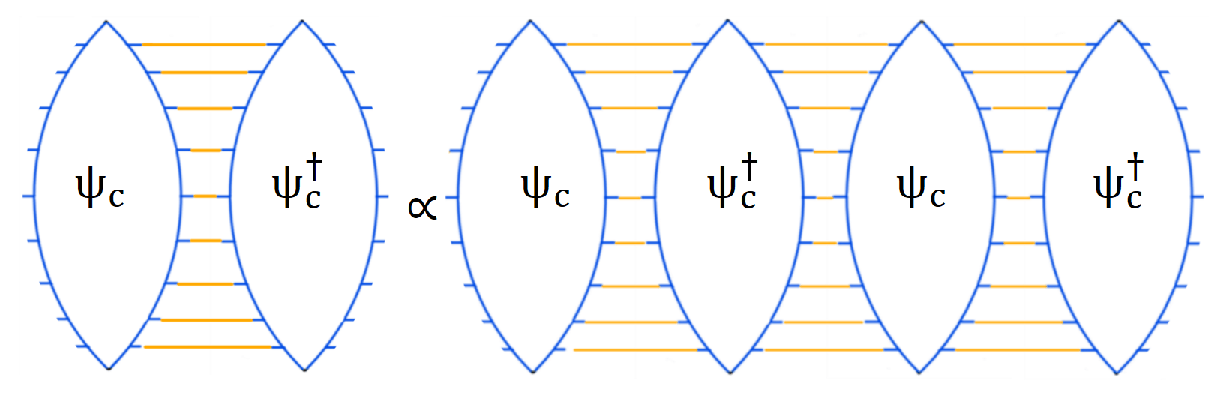}
  \caption{The condition of flat ES.}\label{FigFES}
\end{figure}

The experience one has gained from this picture is that the
thickness of CP region determines whether the ES is flat or not.
Without the boundary effect, the boundary of CP region is composed
of two CP curves, so its thickness is $2d_c$, where
$d_c=\text{arctanh}(\lambda_c)$ is the geodesic distance between
the CP curve (hypercircle) and its axis. Due to the boundary
effect, the CP tensor chain will not be protected any more and the
outer layer of the original CP region will be absorbed by greedy
algorithm. The thickness of such layers is proximately given by
$P$, which is the length of an edge (\ref{lengthp}). So the
thickness of CP region decrease to $2d_c-P$. Since $\Psi_c$ is
protected, (\ref{flatESCP}) is true only if the thickness of the
CP region vanishes.

The evaluation of the geodesic curvature $\lambda$ in a
general tensor network is difficult, which prevents us from
justifying the flatness of ES with CP curvature $\lambda_c$.
Alternatively, this job can be done by calculating CP reduced
interior angle $\kappa_c$, as described in the next subsection.

\subsection{The reduced interior angle of CP tensor chain and ES}\label{SubsectionCPES}

In previous subsections we have shown the relation between the
flatness of ES and the structure of CP tensor chains under the
action of greedy algorithm. In this subsection we show that the
flatness of ES can be justified based on the value of $\kappa_c$.
Specifically, we find that the bigger is $\kappa_c$, the stronger
is the ability of QEC while ES more easily becomes flat, as shown
in Fig.\ref{FigLambdaKappa}. In this figure we further introduce
three quantities which are determined by the $\{b,a\}$ tiling: \be
\kappa_h=\frac a2 - \frac{a-2}2\sqrt{\frac{ab-2a-2b}{ab-2a-2b+4}},
\quad \kappa_1=\frac{b}{b-2},\quad \kappa_0=\frac a2. \ee Because
of (\ref{abconstraint}), the relation $\kappa_h<\kappa_1<\kappa_0$
always holds. If $\kappa_c \in (1, \kappa_h)$, it turns out the
network is not able to implement QEC but has non-flat ES, as
indicated in Fig.\ref{Fig731} and \ref{Fig451}. If $\kappa_c \in
(\kappa_h,\kappa_1)$, then the network can implement QEC and has
non-flat ES, as shown in Fig.\ref{Fig73101} and \ref{Fig45111}. If
$\kappa_c \in [\kappa_1,\kappa_0)$, the ability of QEC will become
stronger but the ES will become ``mixed"
, as shown in Fig.\ref{Fig7310101} and \ref{Fig452}. Finally,
if $\kappa_c = \kappa_0$, the quality of QEC
becomes better but the ES has to be flat, which is exactly the
property of perfect tensors, as shown in Fig.\ref{Fig7311} and \ref{Fig4522}.

Correspondingly, we may propose a geometric quantity in $H^2$
space which plays a similar role as $\kappa_c$ in tensor network.
This quantity is the geodesic curvature $\lambda_c$ of CP curve.
Given a tiling, $\lambda_c$ can be calculated by using $\kappa_c$.
A schematic relation between $\lambda_c$ and QEC and ES is also
illustrated in Fig.\ref{FigLambdaKappa}. While, we do not have
general expressions for the bounds $\lambda_0$ and $\lambda_1$ so
far, which corresponds to $\kappa_0$ and $\kappa_1$, respectively.
\footnote{The main difficulty probably results from the
specification of an unique CP curve corresponding to a CP tensor
chain. Tensor chains are discrete, while curves are continuous. To
assign an unique curve, we have to impose more conditions such as
requiring that the CP curve has the maximal value of geodesic
curvature, which is difficult to handle in practice for a general
tiling.}

Until now, we have constructed a general framework for tensor
networks with tensor constraints, and developed a generalized
greedy algorithm to describe the property of critical protection.
In the remainder of this paper, we will provide detailed
proofs for the quantitative relation between CP tensor chain and
QEC as well as ES, and finally complete the classification of
tensor networks as illustrated in Fig.\ref{FigLambdaKappa}.

\begin{figure}
    \includegraphics[width=350pt]{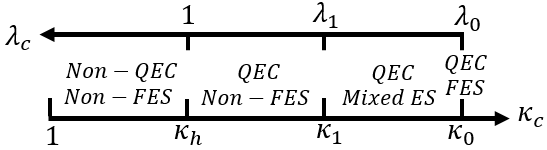}
    \caption{$\lambda_c$ and $\kappa_c$ can be used to classify the property of QEC and ES in tensor network.
    }\label{FigLambdaKappa}
\end{figure}

Those statements can be rephrased into following propositions:
\begin{itemize}
    \item If $\kappa_c=\frac a2$, then $\rho_A$ has flat ES for any choice of $A$;
    \item If $\kappa_c < \frac{a}{2}$, then $\rho_A$ may have non-flat ES for some choices of $A$;
    \item If $\kappa_c\geq\frac{b}{b-2}$, then $\rho_A$ may have flat ES for some choices of $A$;
    \item If $\kappa_c<\frac{b}{b-2}$, then $\rho_A$ has non-flat ES for any choice of {\it large}
    $A$.
\end{itemize}
Now we intend to prove these propositions separately.

\subsubsection{$\kappa_c=\frac a2\Rightarrow$ flat ES}

Based on the discussion in Subsection \ref{subsectionGA}, we
will prove the flatness of ES by showing that any directed cut
appearing in the process of greedy algorithm is unprotected such
that the greedy algorithm will not stop until all the tensors are
absorbed.

To prove a directed cut in the process of greedy algorithm is
unprotected, one need to find out an unprotected tensor chain
connected to the cut. Recall that those directed cuts in greedy
algorithm may have many disconnected components. We firstly
prove a lemma for a cut containing the structure of twigs or
loops, which will greatly simplify the rest of proofs.

\begin{lemma}\label{LemmaLoopCut}
Given a tensor network with $\ke{b,a}$ tiling and $\kappa_c\geq
\frac b{b-2}$, if a directed cut $C$ contains a connected
component whose corresponding sequence of nodes has a form as
$[\cdots,N_0,N_1,\cdots,N_k,\cdots]$ with $N_0=N_k$, or
$(\cdots,N_0,N_1,\cdots,N_k,\cdots)$ with $N_0=N_k$, or
$(N_1,N_2,\cdots,N_k)$, then the cut $C$ is unprotected.
\end{lemma}
\begin{proof}
Denote the sequence of nodes $[N_1,\cdots,N_k]$ as $L_N$. We
assume that all these nodes in $L_N$ are distinct,
otherwise we just replace $N_1$ and $N_k$ by any two nodes which
are identical and the following proof is still valid.

When $k=1$, it is only possible that the connected component is a
single node $N_1$, then the tensor chain $\bbm 0\\a \ebm$ on
$N_1$ is connected to $C$. Since $\bbm 0\\a \ebm\in S_D$, $C$ is
unprotected.

When $k=2$, the shape of $L_N$ is a twig and $N_1$ is the endpoint
of the twig. Step tensor chain $\bbm 1\\a-1 \ebm$ on $N_1$ is
connected to $C$, so $C$ is unprotected. For instance, in
Fig.\ref{FigCut}, $N_{27}=N_{29}$ and the sequence
$L_N=[N_{28},N_{29}]$ forms a twig, $N_{28}$ is the endpoint and
$M_s$ at $N_{28}$ is connected to the cut.

When $k\geq3$, those edges between $N_i$ and $N_{i+1}$, and the
edge between $N_1$ and $N_k$ in $L_N$ form a closed polyline,
{\it e.g.,} see the sequence $[N_{10},N_{11},\cdots,N_{16}]$ in
Fig.\ref{FigCut}. We define the region enclosed by the polyline as
$Y$, which consists of $F$ elementary polygons, $E$ edges
 and $V$ nodes (vertices)
which satisfy Euler's formula
    \be\label{EulerFormula}
    F-E+V=1.
    \ee
    Let the reduced interior angle of $Y$ at $N_i$ be $x_i$ for $i\in\ke{1,2,\cdots,k}$. We have
    \bea
    \sum_{i=1}^{k} x_i +k +a(V-k)=2E,\\
    k\leq V, \\
    bF=2E-k.
    \eea
    From above four formulas, we have
    \be\label{AnglePolygon}
    \sum_{i=1}^k x_i \leq \frac b{b-2} (k-2).
    \ee
    Because $x_k\geq1$, we further have
    \be
    \sum_{i=1}^{k-1} x_i \leq \frac b{b-2} (k-2) - 1.
    \ee
Those nodes $[N_1,N_2,\cdots,N_{k-1}]$ form a tensor chain $M=\bbm
*&*&\cdots&*\\n_1&n_2&\cdots&n_{k-1} \ebm$ connected to $C$, where
$n_i=a-1-x_i$ for $i\in\ke{1,2,\cdots,k-1}$. Recall that
$\kappa_c\geq\frac b{b-2}$. Finally,
    \be
    \sum_{i=1}^{k-1} (a-1-n_i) \leq \kappa_c (k-2) - 1<\kappa_c(k-1)-1.
    \ee
From Theorem \ref{theoremkappac}, $M$ is unprotected and thus
$C$ is unprotected.

In conclusion, when $\kappa_c\geq \frac b{b-2}$, any cut $C$
containing twigs or loops must be unprotected.
\end{proof}

Obviously, $\kappa_c=\frac a2> \frac{b}{b-2}$. Lemma
\ref{LemmaLoopCut} is applicable to this case and those branches
forming twigs or loops in a cut will be absorbed by the greedy
algorithm. Taking the cut in Fig.\ref{FigCut} as an example,
we claim that those nodes in $\ke{N_8,N_9,\cdots,N_{17}},
\ke{N_{27},N_{28},N_{29}}$, and $\ke{N_1',N_2',\cdots,N_7'}$ will
be absorbed.

As a result, now we can focus on the case that the cut $C$ is
single connected and bounded by $\partial A$, which is denoted
as $C=[N_1,N_2,\cdots,N_l]$. Furthermore, these nodes in $C$ are
distinct.

With any choice of single interval $A$ on the boundary of a tensor
network $\Psi$, we apply the greedy algorithm starting from $A$
and from $\bar A$ simultaneously. So two cuts, $C_A$ and $C_{\bar
A}$, appear in $\Psi$ at the same time. We will prove with
mathematical induction that when $\kappa_c=\frac a 2$, either of
these two cuts is unprotected until all the tensors are absorbed.

\begin{figure}
    \centering
    \includegraphics[height=300pt]{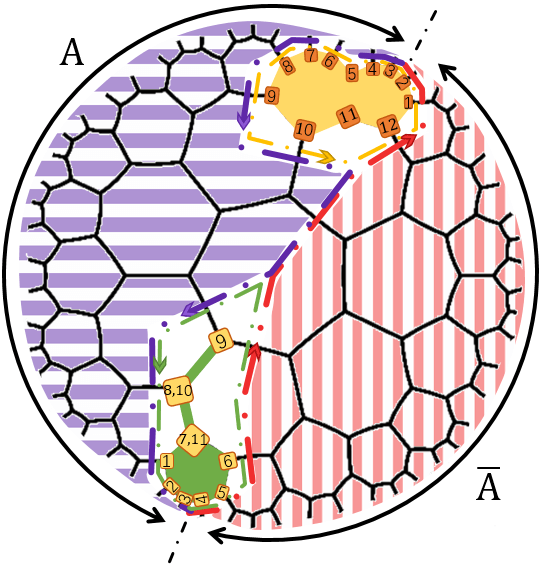}
\caption{The ``directed sum'' of two directed cuts $C_A$ (purple)
and $C_{\bar A}$ (red) consists of two directed closed curves
$C_I$ (green) and $C_I'$ (yellow). The region $I$ corresponding to
$C_I$ is filled in green.}\label{FigDirectedSum}
\end{figure}

Now we consider the configuration of $C_A$ and $C_{\bar A}$. Both
of them are connected to $\partial A$. Besides, they may overlap
at some place, where their directions are opposite, as illustrated
in Fig.\ref{FigDirectedSum}. Then, we define the ``directed sum''
of $C_A$ and $C_{\bar A}$ as the union of them but excluding their
overlapped parts. The directed sum consists of one or more closed
curves, as shown in Fig.\ref{FigDirectedSum}. Select one of them
and denote it as $C_I$, which is a directed cut as well. Set the
sequence of nodes corresponding to $C_I$ to be
$(N_1,N_2,\cdots,N_k)$. We connect these nodes
$(N_1,N_2,\cdots,N_k)$ with edges in order and enclose a region
$I$, which is a union of elementary polygons and edges. At node
$N_i$, let the reduced outer angles of $I$ be $y_i$ and let the
number of edges cut by $C_I$ be $n_i$. Obviously, $y_i=n_i+1$.
Gauss-Bonnet theorem tells that \be\label{GBT1} \sum_{i=1}^k
(\frac{2\pi}{a}y_i-\pi)-2\pi= \text{Area}(H)\geq0, \ee then
\be\label{GBT2} \sum_{i=1}^k n_i\geq \frac{a-2}2 k+a. \ee

Obviously, those edges cut by $C$ are divided into two parts, one
part is cut by $C_A$ and the other is cut by $C_{\bar A}$.
Without loss of generality, we suppose that $C_A$ runs from $N_1$
to $N_{u+1}$. Moreover, $l_1$ edges of $N_1$ are cut by $C_A$ and
$\bar l_{k+1}$ edges cut by $C_{\bar A}$. While for $N_{u+1}$,
$l_{u+1}$ edges are cut by $C_A$ and $\bar l_{u+1}$ edges are cut
by $C_{\bar A}$. Obviously, $l_1+\bar l_{k+1}=n_1$ and
$l_{u+1}+\bar l_{u+1}=n_{u+1}$. We further define that $l_i=n_i$
for $i=2,3,\cdots,u$ and $\bar l_i=n_i$ for $i=u+2,u+3,\cdots,k$.
Then we know that tensor chain $M_A=\bbm *&*&\cdots&*\\
l_1&l_2&\cdots&l_{u+1} \ebm$ is connected to $C_A$ and tensor
chain $M_{\bar A}=\bbm *&*&\cdots&*\\ \bar l_{u+1}&\bar
l_{u+2}&\cdots&\bar l_{k} \ebm$ is connected to $C_{\bar A}$. From
(\ref{GBT2}), \be \sum_{i=1}^{u+1} l_i+\sum_{i=u+1}^{k+1} \bar
l_i=\sum_{i=1}^k n_i \geq \frac{a-2}2 k+a = (u+1)\frac{a-2}{2}+1 +
(k-u+1)\frac{a-2}{2}+1. \ee So, \be \sum_{i=1}^{u+1} l_i\geq
(u+1)\frac{a-2}{2}+1 \quad\text{or}\quad \sum_{i=u+1}^{k+1} \bar
l_i\geq (k-u+1)\frac{a-2}{2}+1, \ee {\it i.e.} \be
\sum_{i=1}^{u+1} (a-1-l_i)\leq (u+1)\kappa_c - 1
\quad\text{or}\quad \sum_{i=u+1}^{k+1} (a-1-\bar l_i) \leq
(k-u+1)\kappa_c - 1. \ee From Theorem \ref{theoremkappac}, either
of $M_A$ or $M_{\bar A}$ is unprotected, so either of $C_A$ or
$C_{\bar A}$ is unprotected. Thus the greedy algorithm will keep
going on until $\text{Area}(H)=0$ at least, which means  two cuts
$C_A$ and $C_{\bar A}$ are overlapped such that all tensors are
absorbed. Then the ES is flat.

\subsubsection{$\kappa_c < \frac{a}{2}\Rightarrow\exists$ non-flat ES}

Next we intend to prove when $\kappa_c < \frac{a}{2}$, there
exists non-flat ES for some choices of single interval $A$ on the
boundary.

Thanks to Theorem \ref{theoremkappac}, when $a$ is odd, we can
specifically choose $A$ on the boundary such that the structure as
shown in Fig.\ref{MES12} is protected under the action of the
greedy algorithm starting from either side. A tensor network
with a special choice for $A$ is shown in Fig.\ref{MES11},
where the structures enclosed by dashed red circles are
protected and prevent the ES from being flat.

When $a$ is even, similarly one can choose $A$ appropriately
such that the structure as shown in Fig.\ref{MES13} is protected.
Then the ES is non-flat.

We remark that such kind of protected structures is common in
    tensor networks, especially when the network is large enough. So
we intend to argue that when $\kappa_c < \frac{a}{2}$, most
    choices of interval $A$ will lead to non-flat ES.

\begin{figure}
    \newcommand{\length}{120pt}
    \newcommand{\distance}{0pt}
    \centering
    \subfigure[]{\label{MES12}
        \includegraphics[height=\length]{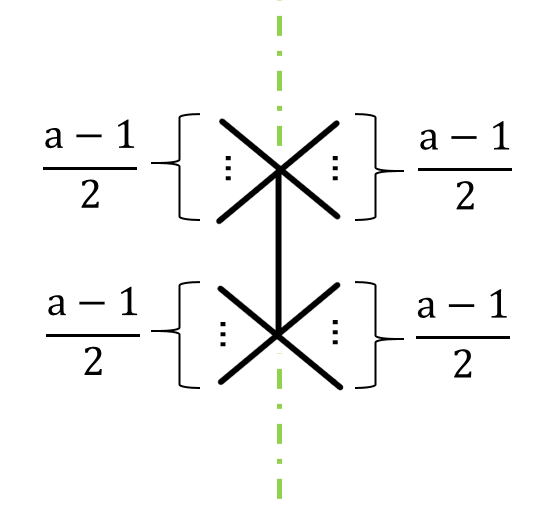}}
    \hspace{\distance}
    \subfigure[]{\label{MES13}
        \includegraphics[height=\length]{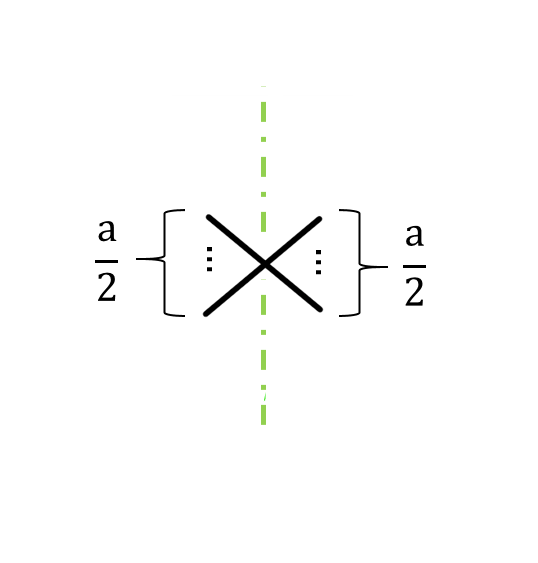}}
    \hspace{\distance}
    \subfigure[]{\label{MES11}
        \includegraphics[height=220pt]{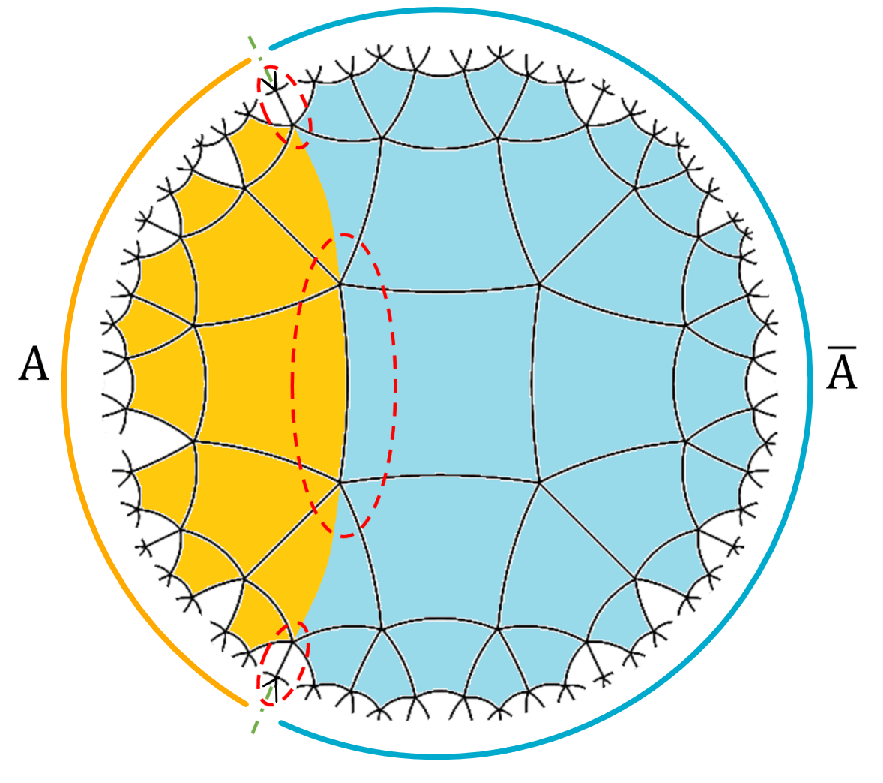}}
    \caption{The green dashed lines in (a) and (b) are boundaries which divide edges of the tensors into two parts; The green dashed lines in (c) are the boundaries which divide the tensor network into two parts. In (c) $\frac{a-1}{2}=2$ and the structure in the red circle cannot be contracted}
    \label{MES1*}
\end{figure}

\subsubsection{$\kappa_c\geq\frac{b}{b-2}\Rightarrow \exists$ flat ES}

In previous subsection we have learned that when $\kappa_c <
\frac{a}{2}$, ES being non-flat is a common phenomenon.
Nevertheless, we point out that when $\kappa_c\geq\frac{b}{b-2}$,
it is possible to construct single interval whose ES is flat.

Next we just prove the existence of flat ES by constructing a
specific interval $A$ with ``minimal secant geodesic'', which is
obtained by following steps (Fig.\ref{FigCut}). We start from the
midpoint of an edge between two uncontracted edges on the
boundary, then connect this point with the midpoint of another
edge in the polygon which has the farthest distance to this point.
Next we choose the neighboring polygon of this new midpoint in the
bulk and connect the midpoint with the other farthest midpoint in
this polygon. Repeat above steps until it reaches the boundary of
the network. The trajectory forms a geodesic called minimal secant
geodesic, denoted by $G_m$. It should be noticed that for a
polygon with odd edges, there are two middle points which are the
farthest from the specified midpoint, one to the left and the
other to the right, as shown in Fig.\ref{FigMES}. We need choose
these two midpoints by turn in above steps, as shown in
Fig.\ref{FigCut}.
\begin{figure}
    \centering
    \includegraphics[height=120pt]{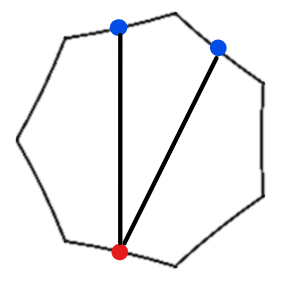}
    \caption{In $b=7$, so for a given red midpoint there are two blue midpoints which are the farthest.}
    \label{FigMES}
\end{figure}
A minimal secant geodesic $G_m$ divides the boundary of network into two
parts $A$ and $\bar A$, which almost have the same size.

We will show that for such a division, the corresponding ES is
flat by proving that the greedy algorithm starting from either
$A$ or $\bar A$ does not stop until the sequence of cuts reaches
$G_m$. The proofs for $A$ and $\bar A$ are parallel. So we only
prove the case for $A$.

Similarly, thanks to Lemma \ref{LemmaLoopCut}, we focus on the
case that the cut $C$ is single connected and connected to
$\partial A$, which is denoted as $C=[N_1,N_2,\cdots,N_l]$,
with distinct nodes.

We give $G_m$ a direction such that $\Phi$ is on its right
hand side. Then $G_m$ becomes a directed cut which is denoted as
$[N_1',N_2',\cdots,N_m']$. By definition, these nodes are
distinct.

When $C$ and $G_m$ are not overlapped, the edges connecting
those nodes in $C$ and $G_m$ at least form a polygon. In general,
they may enclose one or more polygons, as illustrated in
Fig.\ref{FigCut}.

We pick out any one of them and label it as $Y$. Let the set
of those nodes on the boundary of $Y$ to be the union of
$[N_{p+1},N_{p+2},\cdots,N_{p+u}]$ in $C$ and
$[N_{q+1}',N_{q+2}',\cdots,N_{q+v}']$ in $G_m$. $N_{p+1}$ and
$N_{q+1}'$ are neighboring to each other. $N_{p+u}$ and
$N_{q+v}'$ are neighboring to each other. We naturally have
$u\geq2$ after excluding the cases in Lemma \ref{LemmaLoopCut}.
Let the reduced interior angle of $Y$ at $N_{p+i}$ as $x_i$ for
$i\in\ke{1,2,\cdots,u}$ and the reduced interior angle of $Y$ at
$N_{q+j}$ as $x_j'$ for $j\in\ke{1,2,\cdots,v}$. Similar to the
relation in (\ref{AnglePolygon}) in the proof of Lemma
\ref{LemmaLoopCut}, for $Y$, we have
\be\label{EulerFormula2}
    \sum_{i=1}^{u} x_i + \sum_{j=1}^{v} x_j' \leq \frac
    b{b-2}(u+v-2).
\ee Suppose that the part $[N_{q+1}',N_{q+2}',\cdots,N_{q+v}']$
crosses $w$ elementary polygons. Due to the special construction
of $G_m$, we have the relation \be
    v < \frac w2 (b-2)+2.
\ee
So, we have
\be
    \sum_{j=1}^{v} x_j'
    =v+w-1
    > \frac b{b-2}(v-2)+1.
\ee
Plugging it into (\ref{EulerFormula2}), we obtain
\be\label{Cutleq}
    \sum_{i=1}^{u} x_i
    < \frac b{b-2}u-1
    \leq   \kappa_c u - 1.
\ee
On $[N_{p+1},N_{p+2},\cdots,N_{p+u}]$, tensor chain $M=\bbm
*&*&\cdots&*\\n_1&n_2&\cdots&n_u \ebm$  is connected to $C$, where
$n_i=a-1-x_i$ for $i\in\ke{1,2,\cdots,u}$. From (\ref{Cutleq}) and
Theorem \ref{theoremkappac}, we know $M$ is unprotected, thus $C$
is unprotected.

In conclusion, once $C\neq G_m$, $C$ is unprotected and the greedy
algorithm progresses. So those tensors between $A$ and $G_m$
will be absorbed. In parallel, those tensors between $\bar A$
and $G_m$ will be absorbed under the greedy algorithm starting
from $\bar A$. Finally, the sequence of cuts reaches $G_m$,
leading to a flat ES.

\subsubsection{$\kappa_c<\frac{b}{b-2}\Rightarrow$ non-flat ES}

Here we prove that when $\kappa_c<\frac{b}{b-2}$, the ES of a
single and large interval $A$ is non-flat. Perhaps this argument
is the most important part in this section because it supplies us
a quantitative criteria to justify if a tensor network has a
not-flat ES.

Consider a single interval $A$ and its complement $\bar A$ on the
boundary of a given tensor network. There exists a continuous
line, called $G$, connecting two ending points of $A$ with a
minimal cuttings on the edges of the network. The line $G$ divides
the whole network into two sub tensor networks (see Fig.\ref{FigNFES}).

It is noticed that the nearest neighboring tensors of line $G$
form two tensor chains. We call these two tensor chains as $M_A$
and $M_{\bar A}$, respectively. As an example, the skeletons of
these two tensor chains are marked in Fig.\ref{FigNFES}. We set
all the indexes associated with the edges cut by line $G$ as upper
indexes, while the other indexes are lower indexes.

Assume that $M_A$ has $k_A$ nodes, and $M_{\bar A}$ has $k_{\bar
    A}$ nodes. Set the number of elementary polygons crossed by
line $G$ to be $F$. Then we have two equations
\be\label{NES1}\begin{split}
    &\kappa(M_A)k_A + \kappa(M_{\bar A})k_{\bar A}  =  b F,\\
    &k_A + k_{\bar A}  =  (b-2) F + 2.
\end{split}\ee

Now we provide a proof by contradiction. We assume that the ES
would be flat, then $M_A,M_{\bar A}\in S_D$. According to Theorem
\ref{theoremkappaleq}, we have
\be\label{NES2} \kappa(M_A) \leq
\kappa_c,\quad \kappa(M_{\bar A}) \leq \kappa_c.
\ee
We substitute
(\ref{NES2}) into (\ref{NES1}) and get an inequality as
\be\label{NFES3}
\kappa _c\geq \frac{b F}{(b-2) F+2}.
\ee

To simulate real AdS spacetime, the number of layers in a network
is expected to be large enough. Then for large interval $A$,
$F\gg 1$. Since $\frac{2\kappa_c}{b-(b-2)\kappa_c}$ is a finite
number, \be\label{NFES4} F> \frac{2\kappa_c}{b-(b-2)\kappa_c}. \ee
From (\ref{NFES3}) and (\ref{NFES4}), we have
$\kappa_c\geq\frac{b}{b-2}$, contradictory to the initial
assumption. Thus, when $\kappa_c<\frac{b}{b-2}$, the ES of a
single interval $A$ must be non-flat in a network with large
layers.

\begin{figure}
    \centering
    \includegraphics[width=150pt]{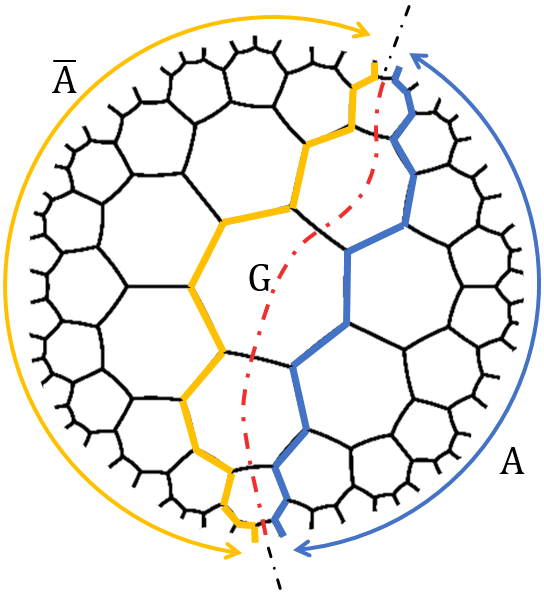}
\caption{The blue (orange) line is interval $A$ ($\bar A$). The
red dashed line is line $G$ with minimal cuttings. The blue
(orange) polyline corresponds to tensor chain $M_A$ ($M_{\bar
A}$).}\label{FigNFES}
\end{figure}

\section{Conclusions and outlooks}

In this paper we have presented a general framework for tensor
networks with tensor constraints based on the tiling of $H^2$
space. A notion of critical protection based on the tensor chain
has been proposed to describe the behavior of tensor networks
under the action of greedy algorithm. In particular, a criteria
has been developed with the help of the average reduced interior
angle of CP chain such that for a given tensor network the ability
of QEC and the flatness of ES can be justified in a quantitative
manner. We have also demonstrated a lot of examples of tensor
network and discussed their properties of QEC and ES. In general,
once the ability of QEC of a tensor network becomes stronger, then
its ES becomes flat more easily, and vice versa. By contrast, it
is fascinating to notice that AdS spacetime is endowed with these
two holographic features with perfect balance indeed. Currently it
is still challenging to construct tensor networks which could
capture all the holographic features of AdS spacetime. What we
have found in this paper may shed light on this issue. Firstly, we
have learned that the notion of critical protection provides a
description on the limit of information transmission with full
fidelity. In the case that the CP curve $H_c$ is a circle, {\it
    i.e.} $\lambda_c L^2>1$, the information in the interior of $H_c$
can be transmitted to its surface without loss, where we have
restored the AdS radius $L$. While, for a circle $H$ which is
larger than $H_c$, its interior information can not be transmitted
to its surface without loss. So we can say that $H_c$ is the
maximal boundary which can holographically store the interior
information \cite{Flammia:2016xvs,Jacobson:2015hqa}. Thus, for a
tensor network which captures the feature of QEC as AdS space, it
must not contain circular CP curves, which requires $\lambda_c
L^2\leq1$. Furthermore, if we intend to construct a single tensor
network which exhibits both QEC and non-flat ES, it seems that the
tensor networks with $\kappa_c \in (\kappa_h, \kappa_1)$ might
have more likelihood to approach this goal.

Next we address some open issues that should be crucial for one to
explore the role of tensor networks with constraints in
holographic approach. Firstly, because of the chain structure of
tensor constraint, in our present framework we have
    investigated QEC and ES only for a single interval on the
boundary, which is just similar to the setup for hyperinvariant
tensor network in \cite{Evenbly:2017htn}. It is an open question
whether QEC can be realized for multi-intervals on the boundary,
as investigated in network with perfect tensors or random tensors
\cite{Hayden:2016cfa,Yang:2015uoa,Pastawski:2015qua}. Actually,
our preliminary investigation reveals that if the number of
intervals is large enough, it would be very hard to realize QEC
with non-flat ES for multi-intervals, because it involves in
constructing tensor constraint with scales as large as the
entanglement wedge of the multi-intervals, which is rather
complicated. We would like to leave this issue for further
investigation.

Secondly, in order to simulate AdS space, it is desirable to send
the number of layers of tensor network to infinity. Then the area
of its boundary goes to infinity as well. Under this limit, the
treatment on the boundary effect of tensor constraints is subtle.
When the CP curve is a hypercircle with $\lambda_c L^2<1$, it has
a constant distance to the geodesic which is $d_c=L\,
\text{arctanh}(\lambda_c L^2)$. The CP curve is unprotected once
the boundary effect is considered, so the boundary effect scales
as $d_c$, which is independent of the number of the layers. When
$d_c/L$ is small, the boundary effect becomes negligible in this
limit comparing to the infinite area of the boundary. However,
when $d_c/L$ is very large, such as $\kappa_c\to\kappa_h+0$, the
boundary effect can not be neglected.

Finally, we are concerned with the issue how to reproduce the
Cardy-Calabrese formula of Renyi entropy (\ref{CardyFormula})
in the framework of tensor networks. It is known that Renyi
entropy depends not only on the tiling and tensor constraints, but
also on the matrix elements of tensors, such as the elements of
tensor $U$ and $Q$ in Appendix \ref{SectionExistence}. In
addition, we are interested in the possible relation between the
CP curve and the gravity dual of Renyi entropy. In
\cite{Dong:2016fnf}, the $n$th-order holographic Renyi entropy can
be calculated by the area of a cosmic brane$_n$ with tension
$T_n$, namely \be
n^2\partial_n\kc{\frac{n-1}{n}S_n}=\frac{\text{Area}(\text{Cosmic
        Brane}_n)}{4G_N}.
\ee The cosmic brane$_n$ backreacts to the geometry at order
$T_nG_N$ where $G_N$ is the Newton constant. However, if we simply
set $T_nG_N\to0$, all the cosmic branes become probe branes
\footnote{ Notice that the tension is $T_n=\frac{n-1}{4nG_N}$ and
    the product $T_nG_N$ is fixed in AdS/CFT
    \cite{Dong:2016fnf}. However, now the question we are asking is
    which tensor network can mimic AdS/CFT. So we loose $T_nG_N$. }. Then, for a given subsystem on the
boundary, those cosmic branes would have the same area and
flat entanglement spectrum appears. According to Subsection
\ref{SubSectionESCPCurve} in our paper, when $d_c/L$ is large, the
entanglement spectrum becomes non-flat, while when $d_c/L$ is
small, the entanglement spectrum becomes flat. It would be
interesting to explore the possible relation between $T_nG_N$ and
$d_c/L$ in the light of this observation.

\acknowledgments

We are grateful to Long Cheng, Glen Evenbly, Wencong Gan, Muxin
Han, Ling-Yan Hung, Hai Lin, Wei Li, Fuwen Shu, Yu Tian, Menghe
Wu, Xiaoning Wu and Hongbao Zhang for helpful discussions and
correspondence. This work is supported by the Natural Science
Foundation of China under Grant No. 11575195. Y.L. also
acknowledges the support from Jiangxi young scientists (JingGang
Star) program and 555 talent project of Jiangxi Province. Z. Y.
Xian is supported by National Postdoctoral Program for Innovative
Talents BX20180318.

\begin{appendix}

    \section{Hyperbolic geometry in 2 dimensional space}\label{SectionH2}

    \subsection{$SL(2,R)$}
    In
    this section we present a brief review on the geometric property
    of $H^2$ space. Without loss of generality, we choose the radius
    of $H^2$ to be 1. Then the scalar curvature of $H^2$ geometry is
    $-2$. The metric in Poincare coordinate $\{x\in R,0<z<\infty\}$ is
    \be
    ds^2=\frac{dx^2+dz^2}{z^2}.
    \ee
    We define new coordinate $\zeta=x+i z$ to rewrite the
    metric as
    \be
    ds^2=-\frac{4 d\zeta d\zeta^*}{(\zeta-\zeta^*)^2}.
    \ee
    The isometry of $H^2$ geometry is $SL(2,R)$, which means
    the form of the metric is unchanged under the coordinate
    transformation
    \be
    \zeta\to \frac{\alpha\zeta+\beta}{\gamma\zeta+\delta},
    \ee
    where real parameters $\alpha,\beta,\gamma,\delta$ satisfy $\alpha\delta-\beta\gamma=1$.

    \subsection{Curves of constant curvature}
    One key notion that we have frequently used in this paper is the
    curve of constant curvature (CCC) in $H^2$ space.  The geodesic
    curvature of a curve with an affine parameter $s$ is given by \be
    \lambda^\mu=\frac{d^2 x^\mu}{ds^2}+\Gamma^\mu_{\nu\rho} \frac{dx^\nu}{ds}
    \frac{dx^\rho}{ds}.
    \ee

    The curves with $\lambda^\mu=0$ are geodesics in $H^2$
    space. The geodesic distance of any two points with
    coordinates $(x_1,z_1)$ and $(x_2,z_2)$ can be derived as
    $d=\text{arccosh}\frac{(x_1-x_2)^2+z_1^2+z_2^2}{2z_1z_2}$.

    There are three kinds of CCC in $H^2$ space, namely, the
    circle, horocircle and hypercircle, as illustrated in
    Fig.\ref{Figcircles}.

    A circle is a curve whose geodesic distance to a given point
    (the center of the circle) is a constant $r$. The geodesic
    curvature of a circle with radius $r$ is $\lambda=\coth(r)$.

    A horocircle (or horocycle) is a curve whose normal geodesics
    all converge asymptotically to its center in the same
    direction, so it is also called limit circle. The geodesic
    curvature of a horocircle is equal to $1$.

    A hypercircle (or hypercycle) is a curve whose points have the
    same orthogonal distance $d$ from a given geodesic, so it is also
    called equidistant curve. The corresponding geodesic is called
    its axis. The geodesic curvature of a hypercircle is
    \be\label{distant}
    \lambda=\tanh(d).
    \ee Of course, a geodesic is a hypercircle with $d=0$.

    As a summary, one can classify all the CCCs in $H^2$ space by their geodesic curvature
    \be\label{ClassifyCCC}
    \begin{array}{ll}
        \lambda>1,   & \text{circle}    \\
        \lambda=1,   & \text{horocircle} \\
        0\leq\lambda<1,  & \text{hypercircle}    \\
    \end{array}
    \ee

    \begin{figure}
        \includegraphics[width=150pt]{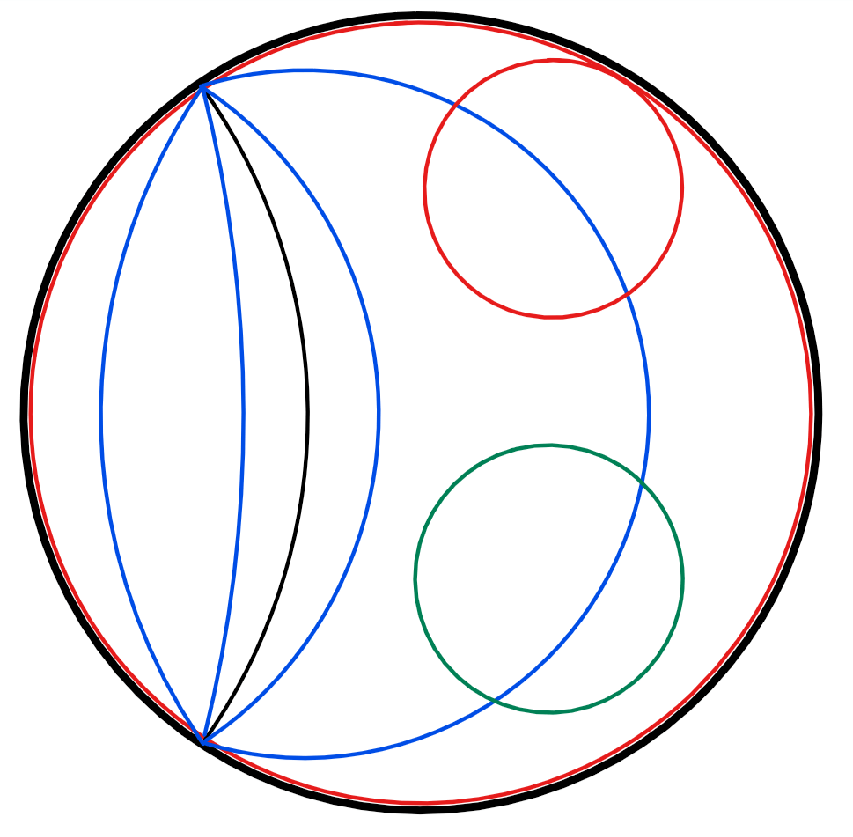}\\
        \caption{Circle (green), horocircles (red) and hypercircles (blue and black) in $H^2$ space. Those hypercircles share the same axis (black). They are symmetric under the reflection with respect to their axis.}\label{Figcircles}
    \end{figure}

    \begin{figure}
        \newcommand{\minipagewidth}{0.3\linewidth}
        \newcommand{\figheight}{110pt}
        \begin{minipage}[t]{\minipagewidth}
            \centering
            \includegraphics[height=\figheight]{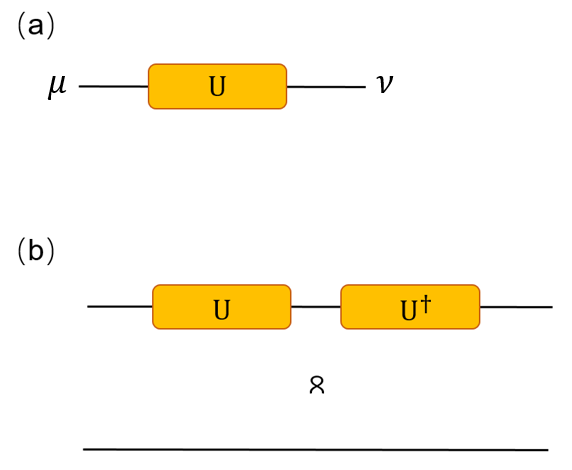}
            \caption{(a) Tensor $U$. (b) Tensor $U$ is proportional to an isometry.}\label{FigU}
        \end{minipage}
        \hfill
        \begin{minipage}[t]{\minipagewidth}
            \centering
            \includegraphics[height=\figheight]{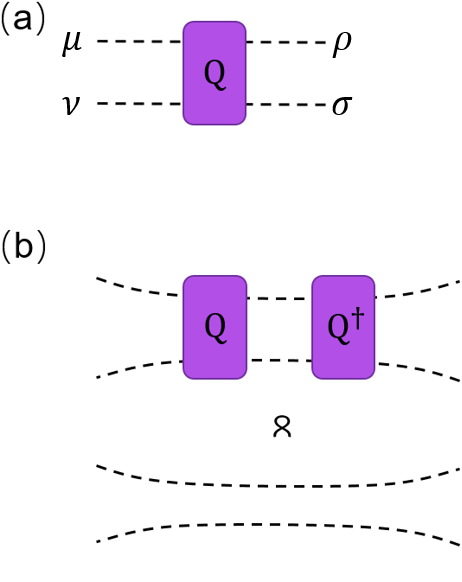}
            \caption{(a) Tensor $Q$, where two indexes on each side are grouped together. (b) Tensor $Q$ is proportional to the isometry between the two grouped indexes.}\label{FigQ}
        \end{minipage}
        \hfill
            \begin{minipage}[t]{\minipagewidth}
            \centering
            \includegraphics[height=\figheight]{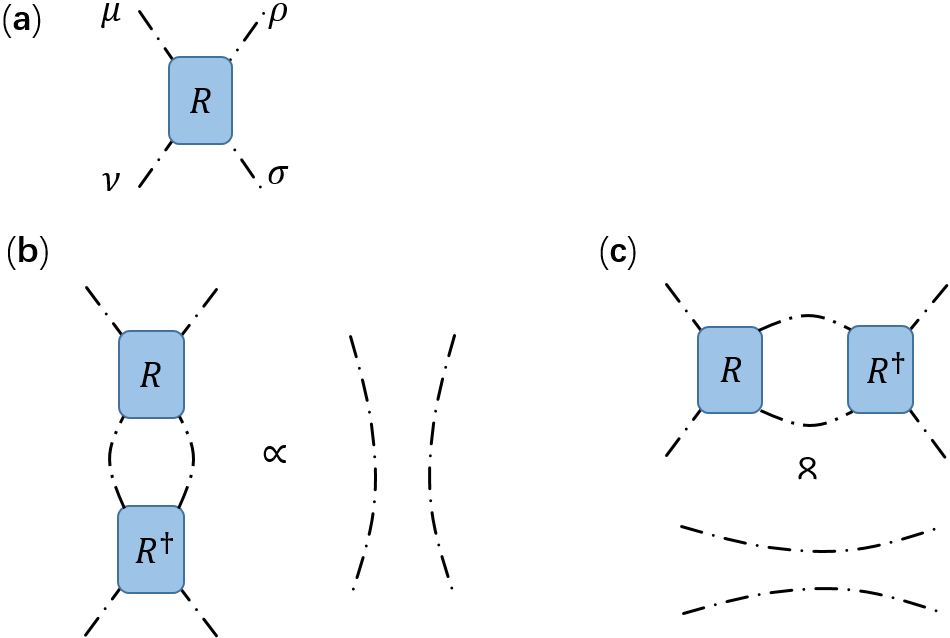}
            \caption{(a) Tensor $R$. (b) Tensor $R$ is proportional to isometries along two directions.}\label{FigR}
        \end{minipage}
    \end{figure}

    \begin{figure}
        \newcommand{\minipagewidth}{0.47\linewidth}
        \newcommand{\figwidth}{120pt}
        \newcommand{\figheight}{120pt}
        \begin{minipage}[t]{\minipagewidth}
            \centering
            \subfigure[]{
                \includegraphics[height=\figheight]{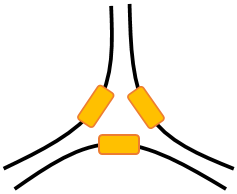}
            }
            \subfigure[]{
            \includegraphics[height=\figheight]{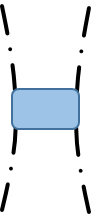}
            }
            \caption{(a) Tensor $T$ and (b) Tensor $E$ for the tensor network with $\ke{7,3}$ tiling and
                $S_c=\left\{\bbm 1\\3 \ebm,\bbm 1&1\\1&1 \ebm\right\}$.}\label{FigTE7311}
        \end{minipage}
        \hfill
        \begin{minipage}[t]{\minipagewidth}
            \centering
            \subfigure[]{
                \includegraphics[height=\figheight]{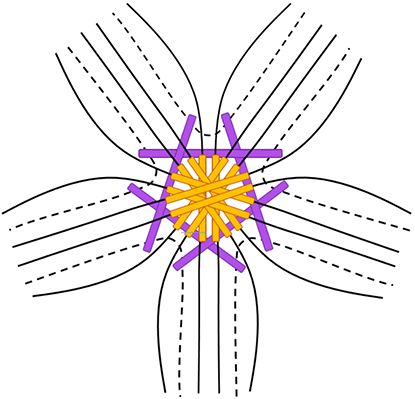}
            }
            \subfigure[]{
                \includegraphics[height=\figheight]{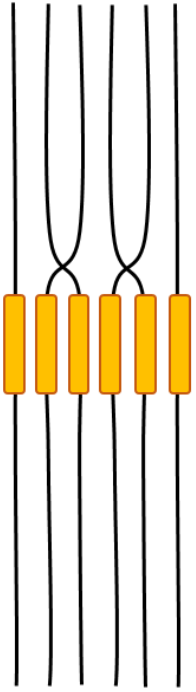}
            }
            \caption{(a) Tensor $T$ and (b) Tensor $E$ for the tensor network with $\{4,5\}$ tiling
                and $S_c=\left\{\bbm 1\\4 \ebm,\bbm 1&1&1\\3&2&3 \ebm
                \right\}$.}\label{FigTE45111}
        \end{minipage}
    \end{figure}

    \section{Specific construction of tensors subject to tensor constraints}\label{SectionExistence}

Analogous to the construction of tensors in \cite{Evenbly:2017htn}
and \cite{Ling:2018qec}, we define tensors $U$, $Q$ and $R$ as
the building blocks for tensors $T$ and $E$, as shown in
Fig.\ref{FigU}, \ref{FigQ} and \ref{FigR}. The elements of tensor
$U$ are $U_{\mu\nu}$, which satisfy
    \be
        U_{\mu\nu}=U_{\nu\mu}, \quad
        \sum_{\nu} U_{\mu\nu}U_{\rho\nu}^*\propto\delta_{\mu\rho}.
    \ee
    The elements of tensor $Q$ are $Q_{\mu\nu\rho\sigma}$ which satisfy
    \be
      Q_{\mu\nu\rho\sigma}=Q_{\rho\sigma\mu\nu}=Q_{\nu\mu\sigma\rho}, \quad
        \sum_{\rho\sigma} Q_{\mu\nu\rho\sigma}Q_{\mu'\nu'\rho\sigma}^*\propto\delta_{\mu\mu'}\delta_{\nu\nu'}.
    \ee
    where two indexes $\mu\nu$ ($\rho\sigma$) are grouped together.
    The elements of tensor $R$ are $R_{\mu\nu\rho\sigma}$, which satisfy
    \be
        R_{\mu\nu\rho\sigma}=R_{\rho\sigma\mu\nu}=R_{\nu\mu\sigma\rho}, \quad
        \sum_{\rho\sigma} R_{\mu\nu\rho\sigma}R_{\mu'\nu'\rho\sigma}^*\propto\delta_{\mu\mu'}\delta_{\nu\nu'}, \quad
        \sum_{\nu\sigma} R_{\mu\nu\rho\sigma}R_{\mu'\nu\rho'\sigma}^*\propto\delta_{\mu\mu'}\delta_{\rho\rho'}.
    \ee

For the tensor network with $\ke{7,3}$ tiling and $S_c=\left\{\bbm
1\\2 \ebm,\bbm 1&1\\1&1 \ebm\right\}$, and the tensor network with
$\{4,5\}$ tiling and $S_c=\left\{\bbm 1\\4 \ebm,\bbm 1&1&1\\3&2&3
\ebm \right\}$, we construct tensors $T$ and $E$ in
Fig.\ref{FigTE7311} and \ref{FigTE45111}, respectively.
The specific structures of the top tensor chain in these two
tensor networks are also shown in Fig.\ref{FigMt}.

    \begin{figure}
        \centering
        \subfigure[]{
            \includegraphics[width=0.37\linewidth]{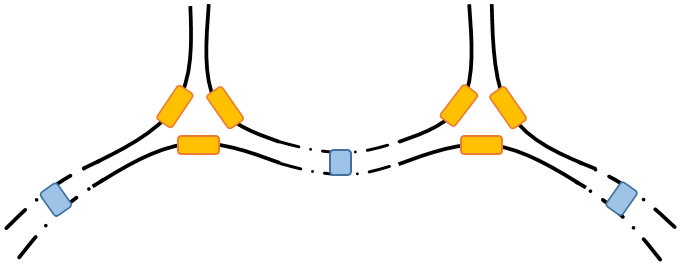}
        }
        \subfigure[]{
            \includegraphics[width=0.58\linewidth]{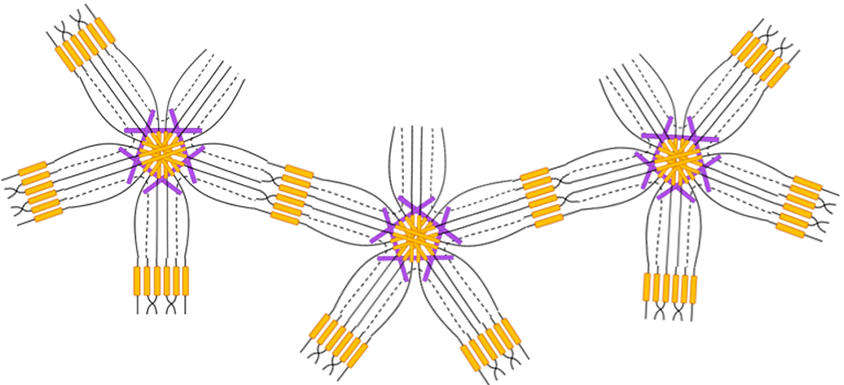}
        }
        \caption{(a) The top tensor chain $\bbm 1&1\\2&2 \ebm$ in the tensor network with $\ke{7,3}$ tiling and $S_c=\left\{\bbm 1\\2 \ebm,\bbm 1&1\\1&1 \ebm\right\}$. (b) The top tensor chain $\bbm 1&1&1\\3&2&3 \ebm$ in the tensor network with $\{4,5\}$ tiling and $S_c=\left\{\bbm 1\\4 \ebm,\bbm 1&1&1\\3&2&3 \ebm \right\}$.
        }\label{FigMt}
    \end{figure}

\end{appendix}

\end{document}